\documentclass[a4paper,11pt]{scrartcl}
\pdfoutput=1

\usepackage[utf8]{inputenc}
\usepackage[T1]{fontenc}
\usepackage{microtype}
\usepackage{lmodern}
\usepackage{enumitem}

\usepackage{tikz,pgfplots}
\usetikzlibrary{automata, positioning, arrows, shapes, arrows.meta}
\usetikzlibrary{calc,decorations.pathmorphing,shapes,petri,positioning,arrows,automata,shapes.geometric}
\usetikzlibrary{graphs, patterns,decorations.pathreplacing,patterns.meta}

\usepackage{thm-restate}
\usepackage{xcolor}
\usepackage{xspace}
\usepackage[fleqn]{amsmath}
\usepackage{amsthm}
\usepackage{mathtools}
\usepackage{amssymb}
\usepackage{stmaryrd}
\usepackage{authblk}
\usepackage{mathbbol}

\DeclareSymbolFontAlphabet{\mathbb}{AMSb}
\DeclareSymbolFontAlphabet{\mathbbl}{bbold}

\newcommand{\N}{\mathbb{N}}

\newcommand{\Z}{\mathbb{Z}}

\definecolor{nicebg}{HTML}{f6f0e4}
\definecolor{nicered}{HTML}{7f0a13}
\definecolor{nicebgred}{HTML}{f2e7e8}
\definecolor{niceblue}{HTML}{104354}
\definecolor{nicebgblue}{HTML}{e8edee}
\definecolor{nicegreen}{HTML}{217516}
\definecolor{nicebggreen}{HTML}{e9f1e8}
\definecolor{nicepurple}{HTML}{884bab}
\definecolor{nicebgpurple}{HTML}{f3edf7}
\definecolor{niceorange}{HTML}{d27c11}
\definecolor{nicebgorange}{HTML}{fbf2e8}
\definecolor{nicepink}{HTML}{e95f9f}
\definecolor{nicebgpink}{HTML}{fdeff6}
\definecolor{niceredlight}{HTML}{c9888d}
\definecolor{nicebluelight}{HTML}{78a4b8}
\definecolor{nicegreenlight}{HTML}{76de68}
\definecolor{nicepurplelight}{HTML}{bc87db}
\definecolor{niceredbright}{HTML}{bd0310}
\definecolor{nicebgredbright}{HTML}{f9e6e8}
\definecolor{nicebluebright}{HTML}{197b9b}
\definecolor{nicebgbluebright}{HTML}{e8f2f5}

\newcommand{\TraNs}{tra}
\newcommand{\TraName}[2][]{\tag*{$\langle${\textnormal {\textsf{#2}}}$\rangle$}\label{\TraNs:#1:#2}}
\newcommand{\TraRef}[2][]{\text{\ref{\TraNs:#1:#2}}}

\newtheorem{theorem}   {Theorem}
\newtheorem{construction}   {Construction}
\newtheorem{definition}  {Definition}
\newtheorem{proposition} {Proposition}
\newtheorem{lemma}       {Lemma}
\newtheorem{example}     {Example}
\newtheorem{corollary}   {Corollary}

\newtheorem{remark} {Remark}

\makeatletter
\RequirePackage[bookmarks,unicode,colorlinks=true]{hyperref}%
   \def\@citecolor{niceblue}%
   \def\@urlcolor{niceblue}%
   \def\@linkcolor{nicered}%

\def\orcidID#1{\smash{\href{http://orcid.org/#1}{\protect\raisebox{1pt}{\protect\includegraphics{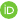}}}}}
\makeatother

\usepackage{etoolbox}
\makeatletter
\pretocmd\start@gather{%
    \if@minipage\kern-\topskip\kern-\baselineskip\kern+7pt\fi
}{}{}
\makeatother

\newcommand{\parag}[1]{\paragraph*{#1.}}

\makeatletter
\renewcommand\paragraph{\scr@startsection{paragraph}{4}{\z@}%
      {.5ex \@plus.25ex \@minus.25ex}%
      {-1em}%
      {\normalfont\normalsize\sffamily\bfseries}}
\renewcommand\subparagraph{\scr@startsection{subparagraph}{5}{\parindent}%
      {0ex \@plus.25ex}%
      {-1em}%
      {\normalfont\normalsize\sffamily\bfseries}}
\makeatother

\usepackage{enumitem}
\setlist[enumerate,1]{itemsep=0pt,topsep=1ex,before={\pagebreak[1]}}
\setlist[itemize,1]{itemsep=0pt,topsep=1ex}

\usepackage[format=plain]{caption}

\csname endofdump\endcsname
\endofdump

\newcommand{\pre}{\textsc{pre}}
\newcommand{\post}{\textsc{post}}
\newcommand{\out}{\textsc{out}}

\newcommand{\supp}[1]{\llbracket #1\kern1pt \rrbracket}
\newcommand{\cfg}[1]{\Lbag #1 \Rbag}

\DeclareMathOperator{\inj}{inj}

\newcommand{\one}[1]{\mathbbl{1}_{#1}}
\newcommand{\pr}{\mathrm{pr}}
\newcommand{\len}{\mathrm{len}}
\newcommand{\zero}{\mathbbl{0}}
\newcommand{\Gcal}{\mathcal{G}}
\newcommand{\Pcal}{\mathcal{P}}

\newcommand{\Wlog}{wlog}

\newcommand{\tref}[1]{{\hyperref[#1]{Theorem~\ref*{#1}}}}
\newcommand{\lref}[1]{{\hyperref[#1]{Lemma~\ref*{#1}}}}
\newcommand{\pref}[1]{{\hyperref[#1]{Proposition~\ref*{#1}}}}
\newcommand{\cref}[1]{{\hyperref[#1]{Corollary~\ref*{#1}}}}
\newcommand{\csref}[1]{{\hyperref[#1]{Construction~\ref*{#1}}}}
\newcommand{\eref}[1]{{\hyperref[#1]{Example~\ref*{#1}}}}
\newcommand{\fref}[1]{{\hyperref[#1]{Figure~\ref*{#1}}}}
\newcommand{\rref}[1]{{\hyperref[#1]{Remark~\ref*{#1}}}}
\newcommand{\sref}[1]{{\hyperref[#1]{Section~\ref*{#1}}}}

\newcommand{\Tower}{\textsc{Tower}}
\newcommand{\InhomTower}{\textsc{InhomTower}}
\newcommand{\InhomTowerCancel}{\textsc{InhomTowerCancel}}
\newcommand{\BigModulo}{\textsc{BigModulo}}
\newcommand{\ModuloCombined}{\textsc{ModuloCombined}}

\makeatletter
\newenvironment{protocol}[1]{%
    \newcommand{\HeadSingle}[1]{%
        \noindent\textbf{\Title.}\hspace*{2mm} & ##1 \\[0.1cm]%
    }%
    \newcommand{\HeadMulti}[1]{\HeadSingle{{%
        \tabcolsep=0pt%
        \def\item{\\\quad}%
        \begin{tabular}[t]{l}%
        ##1%
        \end{tabular}%
    }}}%
    \newcommand{\Head}{\@ifstar{\HeadMulti}{\HeadSingle}}
    \newcommand{\States}{\Init\def\Title{States}\Head}%
    \newcommand{\Input}{\def\Title{Input}\Head}%
    \newcommand{\Output}{\def\Title{Output}\Head}%
    \newcommand{\Transitions}{\def\Title{Transitions}\Head}%
    \def\Init{%
    	\ifvmode%
    	\par%
        \begin{tabular}{ll}%
        \def\Init{}%
    }%
    \def\Done{%
        \end{tabular}%
        \def\Done{}%
    }%
    \par\medskip%
    \noindent\textbf{Protocol }\textsc{#1}\textbf{.}\nopagebreak%
    \par\medskip%
}{%
    \Done%
}
\makeatother

\makeatletter
\newcommand{\Customlabel}[2]{%
   \protected@write \@auxout {}{\string \newlabel {#1}{{#2}{\thepage}{#2}{#1}{}} }%
   \hypertarget{#1}{}
}
\makeatother

\begin{document}
\title{The Black Ninjas and the Sniper: \\ On Robustness of Population Protocols}
\author{Benno Lossin, Philipp Czerner \orcidID{0000-0002-1786-9592}, \\
	Javier Esparza \orcidID{0000-0001-9862-4919}, 
Roland Guttenberg \orcidID{0000-0001-6140-6707},
Tobias Prehn}
\affil{\large \{lossin, czerner, esparza, guttenbe\}@in.tum.de \\
	Department of Informatics, TU München, Germany}
\date{}

\maketitle

\vspace*{-9mm}

\begin{abstract}
\noindent\textbf{Abstract.} Population protocols are a model of distributed computation in which an arbitrary number of indistinguishable finite-state agents interact in pairs to decide some property of their initial configuration.  
We investigate the behaviour of population protocols under adversarial faults that cause agents to silently crash and no longer interact with other agents.  As a starting point, we consider the property  ``the number of agents exceeds a given threshold $t$'', represented by the predicate $x \geq t$, and show that the standard protocol for $x \geq t$ is very fragile: one single crash in a computation with $x:=2t-1$ agents can already cause the protocol to answer incorrectly that $x \geq t$ does not hold. However, a slightly less known protocol is \emph{robust}: for any number $t' \geq t$ of agents, at least $t' - t+1$ crashes must occur for the protocol to answer that the property does not hold. 

\smallskip We formally define robustness for arbitrary population protocols, and investigate the question whether every predicate computable by population protocols has a robust protocol. Angluin et al. proved in 2007 that population protocols decide exactly the Presburger predicates, which can be represented as Boolean combinations of threshold predicates of the form  $\sum_{i=1}^n a_i \cdot x_i \geq t$ for $a_1,...,a_n, t \in \mathbb{Z}$ and modulo prdicates of the form  $\sum_{i=1}^n a_i \cdot x_i \bmod m \geq t $ for $a_1, \ldots, a_n, m, t \in \mathbb{N}$. We design robust protocols for all threshold and modulo predicates. We also show that, unfortunately, the techniques in the literature that construct a protocol for a Boolean combination of predicates given protocols for the conjuncts do not preserve robustness. So the question remains open.

\end{abstract}

\def\phi{\varphi}

\hfill \textit{For Joost-Pieter, Honorary Sensei of the Black Ninjas.}

\section{Introduction}

\newcommand{\Yes}{\textsc{y}}
\newcommand{\No}{\textsc{n}}
%\vspace{6pt}
%\begin{tabular}{p{70pt}l}
%  &\emph{{\Large``}\,\,You can't see her in the night.} \\
%  &\emph{{\Large\phantom{``}}\,\,She's a black ninja.\,\,{\Large''}}
%  \hfill --- \small Battle Beast {\scriptsize (2013)} \\[6pt]
%\end{tabular}

The Black Ninjas were used in \cite{BlondinEJK18} to provide a gentle introduction to population protocols, a model of distributed computation introduced by Angluin et al. in \cite{AngluinADFP06} and since then very much studied. We quote (with some changes) the first lines of the introduction to \cite{BlondinEJK18}:

\begin{quote}
The Black Ninjas are an ancient secret society of warriors, so secret that its members do not even know each other or how many
they are. When there is a matter to discuss, Sensei, the founder of the society, sends individual messages to each ninja, asking them to meet.

As it happens, all ninjas have just received a note asking them to meet in a certain Zen garden at midnight, wearing their black uniform,
in order to attack a fortress of the Dark Powers.  When the ninjas reach the garden in the gloomy night, 
the weather is so dreadful that it is impossible to see or hear anything at all. This causes a problem, because
the ninjas should only attack if they are at least 64, and there are always no-shows: Some ninjas are wounded, others are under evil spells, and others still have tax forms to fill. 
\end{quote}
Is there a way for the ninjas to find out, despite the adverse conditions,  if at least 64 of them have shown up?

\medskip\paragraph{A first protocol.}

Yes, there is a way. Sensei  has looked at the weather forecast, and the note sent to the ninjas contains detailed instructions on how to proceed after they reach the garden.  Each ninja must bring pebbles in their pockets, and a pouch initially containing one pebble. When they reach the garden, they must start  to wander \emph{randomly} around the garden.  When two ninjas with $n_1$ and $n_2$ pebbles in their pouches  happen to bump into each other, they compute $n_1+n_2$ and proceed as follows:
\begin{itemize}
  \item If $n_1 + n_2 < 64$, then one ninja gives her pebbles to the other, that is, after the encounter their pouches contain $n_1+n_2$ and $0$ pebbles, respectively.
  \item If $n_1 + n_2 \geq  64$, both ninjas put 64 pebbles in their pouches.
\end{itemize}

Formally, at every moment in time each ninja is in one state from the set $Q :=\{ 0, 1, \ldots, 64 \}$, representing the number of pebbles in her pouch. The possible transitions are
$$\begin{array}{rcllll}
n_1, n_2 & \mapsto & 
\begin{cases} 
n_1+n_2, 0 &   \text{ if  } n_1 + n_2 < 64 \\
64, 64 &   \text{ otherwise }
\end{cases}
\end{array}$$
The protocol can be visualized as a Petri net with one place for each state, and one Petri net transition for each transition of the protocol. The Petri net representation of the protocol for attack threshold 4 instead of 64 is shown on the left of \fref{fig:senseis-protocols}.

\begin{figure}[t]
    \centering
    \colorlet{colFalse}{magenta!60!red}
    \colorlet{colTrue}{cyan!60!blue}
    \colorlet{colFalseDark}{purple!80!black}
    \colorlet{colTrueDark}{blue}
    \begin{tikzpicture}[->, node distance=0.75cm, auto, thick, transform shape, scale=0.75]
        % Angluin et al
        \node[transition] (ft11) at (0, 0) {};
        \node[tokens=2,place,niceredbright,fill=niceredbright!20,label={0:\textbf{{\Large 0}}}] (f0) [above = of ft11] {};
        \node[transition] (ft12) [above = of f0] {};
        \node[tokens=1,place,niceredbright,fill=niceredbright!20!white,label={0:\textbf{{\Large 3}}}] (f3) [above = of ft12] {};
        \node[tokens=2,place,niceredbright,fill=niceredbright!20!white,label={0:\textbf{{\Large 1}}}] (f1) [left = of f3] {};
        \node[place,niceredbright,fill=niceredbright!20!white,label={0:\textbf{{\Large 2}}}] (f2) [right = of f3] {};
        \node[transition] (ft33) [above = of f3] {};
        \node[transition] (ft22) [above = of f2] {};
        \node[transition] (ft23) [left = 0.37cm of ft22] {};
        \node[transition] (ft13) [above = of f1] {};
        \node[place,nicebluebright,fill=nicebluebright!20!white,label={90:\textbf{{\Large 4}}}] (f4) [above = 3cm of f3] {};

        \node (u24) [right = 0.9cm of f4] {};
        \node (u34) [left = 0.9cm of f4] {};

        \node[transition,nicebluebright] (ft14) [above = 0.5cm of u34] {};
        \node[transition,nicebluebright] (ft04) [above = 0.5cm of u24] {};
        \node[transition,nicebluebright] (ft24) [below = 0.5cm of u24] {};
        \node[transition,nicebluebright] (ft34) [below = 0.5cm of u34] {};

        \path[->, thick]
            (ft11) edge (f0)
            (ft12) edge (f0)
        ;
        \path[->, thick]
            (f1)    edge[bend right] node[left] {$2$}      (ft11)
            (ft11)  edge[bend right]                        (f2)
            (f1)    edge                                        (ft12)
            (f2)    edge                                   (ft12)
            (ft12)  edge                                   (f3)
            (f2)    edge node[right] {$2$}             (ft22)
            (ft22)  edge node[right] {$2$}        (f4)
            (f1)    edge                                          (ft13)
            (f3)    edge                                     (ft13)
            (ft13)  edge node[right] {$2$}            (f4)
            (f2)    edge                                     (ft23)
            (f3)    edge                                     (ft23)
            (ft23)  edge node[right] {$2$}                          (f4)
            (f3)  edge node[left] {$2$}                          (ft33)
            (ft33)  edge node[left] {$2$}                          (f4)
        ;
        \path[->, thick, nicebluebright]
            (f0) edge[out=25,in=-60] (ft04)
            (f1) edge[bend left] (ft14)
            (f2) edge[bend right] (ft24)
            (f3) edge (ft34)
            (ft04) edge[bend right=15] node[above] {$2$} (f4)
            (f4) edge[bend right=15] (ft04)
            (ft14) edge[bend left=15] node[above] {$2$} (f4)
            (f4) edge[bend left=15] (ft14)
            (ft24) edge[bend right=15] node[above] {$2$} (f4)
            (f4) edge[bend right=15] (ft24)
            (ft34) edge[bend left=15] node[left] {$2$} (f4)
            (f4) edge[bend left=15] (ft34)
            ;

        % Tower
        \node[tokens=2,place,niceredbright,fill=niceredbright!20!white,label={0:\textbf{{\Large 1}}}] (t1) at (7, 0) {};
        \node[transition] (tt1) [above = of t1] {};
        \node[tokens=2,place,niceredbright,fill=niceredbright!20!white,label={0:\textbf{{\Large 2}}}] (t2) [above = of tt1] {};
        \node[transition] (tt2) [above = of t2] {};
        \node[tokens=1,place,niceredbright,fill=niceredbright!20!white,label={0:\textbf{{\Large 3}}}] (t3) [above = of tt2] {};
        \node[transition] (tt3) [above = of t3] {};
        \node[place,nicebluebright,fill=nicebluebright!20!white,label={90:\textbf{{\Large 4}}}] (t4) [above = of tt3] {};
        
        \node (u24) [right = 0.9cm of t4] {};
        \node (u34) [left = 0.9cm of t4] {};

        \node[transition,nicebluebright] (ac1) [above = 0.5cm of u34] {};
        \node[transition,nicebluebright] (ac2) [below = 0.5cm of u24] {};
        \node[transition,nicebluebright] (ac3) [below = 0.5cm of u34] {};

        \path[->, thick]
            (t1)    edge[bend right] node[right] {$2$} (tt1)
            (tt1)   edge                               (t2)
            (tt1)   edge[bend right]                   (t1)
            (t2)    edge[bend right] node[right] {$2$} (tt2)
            (tt2)   edge                               (t3)
            (tt2)   edge[bend right]                   (t2)
            (t3)    edge[bend right] node[right] {$2$} (tt3)
            (tt3)   edge                               (t4)
            (tt3)   edge[bend right]                   (t3)
            ;
        \path[->, thick, nicebluebright]
            (t1) edge[bend left] (ac1)
            (t2) edge[bend right] (ac2)
            (t3) edge (ac3)
            (ac1) edge[bend left=15] node[above] {$2$} (t4)
            (t4) edge[bend left=15] (ac1)
            (ac2) edge[bend right=15] node[above] {$2$} (t4)
            (t4) edge[bend right=15] (ac2)
            (ac3) edge[bend left=15] node[left] {$2$} (t4)
            (t4) edge[bend left=15] (ac3)
            ;
    \end{tikzpicture}
\caption{Petri nets for the first (left) and second (right) protocols with attack threshold 4. Accepting states are shaded in blue, rejecting states in red. The number of pebbles (left) or the level (right) of a state are written next to its corresponding place in boldface. The tokens in the places show configurations reachable from the initial configuration with $5$ ninjas. The accumulation transitions are colored in blue.}
    \label{fig:senseis-protocols}
\end{figure}
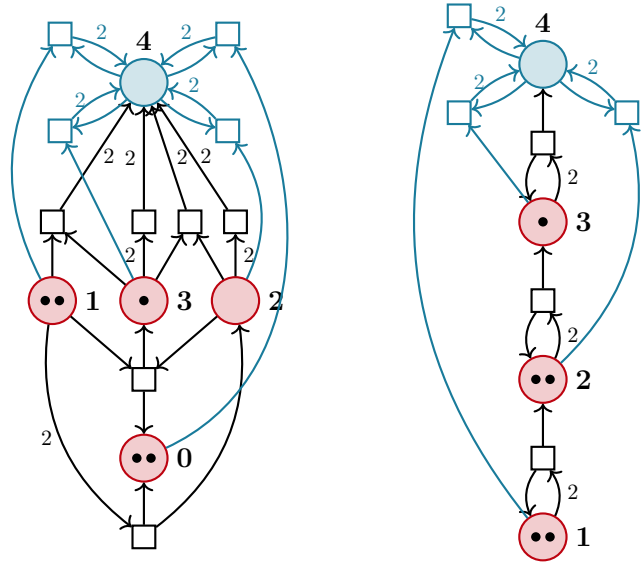

Sensei reasons as follows: If at least 64 ninjas show up, then eventually two ninjas with at least 64 pebbles between them interact. These ninjas move to state 64, and from then on word quickly spreads, and eventually all ninjas are in state 64, and stay there forever. If, on the contrary, less than 64 ninjas show up, then no ninja ever reaches state 64. In both cases, the ninjas eventually reach a \emph{consensus}: all of them are in state  64, or none is, and that consensus does not change, it is \emph{stable}.  Further, since Sensei knows that the society consists of 200 ninjas---a fact the ninjas themselves ignore---and she is good at math, she computes that after 10 minutes the consensus has been reached with high probability (ninjas move very fast!). So her final instruction to the ninjas is: Execute the protocol for 10 minutes, and then attack if you are in state 64.

\medskip\subsection{The sniper and a second protocol.}
The ninjas use the protocol to successfully  attack several fortresses of the Evil Powers. However, a spy finds out the point of their next gathering, and the Evil Powers send a sniper equipped with a night vision device and a powerful gun whose bullets can silently vaporize a ninja leaving no trace. No less than 126 ninjas show up at the gathering, but the sniper manages to prevent the attack by eliminating one single ninja! The sniper lets the ninjas wander around, executing the protocol, until they reach a configuration in which two ninjas have pouches with 63 pebbles each, and all other pouches are empty. Right before these two ninjas bump into each other, the sniper  eliminates one of them. The other ninjas, unaware of this, are left in a configuration in which one ninja has 63 pebbles in her pouch, and the pouches of the other 124 ninjas are empty. The ninjas keep interacting, but none ever reach state 64.

When Sensei finds out about the sniper, she initially despairs. But then she realizes: munitions are expensive, and the sniper can carry only a few rounds, so perhaps there is a remedy. So she searches for a \emph{robust} protocol. If $n \geq 64$ ninjas reach the garden, the sniper can only prevent an attack by eliminating  more than $n - 64$ ninjas.  Clearly, this is the best she can achieve, because if the sniper succeeds in removing more than $n-64$ ninjas, then the ninjas should not attack anyway.

%At first she searches for a \emph{self-stabilizing} protocol in which, after the sniper runs out of munition, the ninjas always reach the consensus corresponding to how many are left. However, she quickly notices that no such protocol can exist. For this, she considers the situation in which 64 ninjas show up, and the sniper lets them proceed until all of them have reached stable consensus to attack. At that point, the sniper eliminates one ninja. The 63 remaining ninjas do not known that one of them has been eliminated, and so they continue executing the protocol. But any finite sequence of interactions the 63 ninjas can execute can also take place if no ninja had been eliminated. So the consensus cannot change. 

Sensei is a capable leader, and she finds a robust protocol! The states are now the numbers 1 to 64. Sensei visualizes these numbers as the floors or levels of a tower with 64 levels, and visualizes a ninja in state $\ell$ as occupying the $\ell$-th level. 
Initially all ninjas are in the first level. When two ninjas bump into each other, they interact as follows:
\begin{itemize}
\item if they are both in the same level, say the $\ell$-th, and $\ell < 64$, then one of them moves to level $\ell + 1$, and the other stays at level $\ell$. 
\item if one of them is in level 64, then the other moves to level 64 too (if not yet there);
\end{itemize}

\noindent Formally, at every moment in time each ninja is in a state from the set $Q :=\{ 1, 2, ..., 64\}$. The possible transitions are
$$\begin{array}{rcllll}
\ell_1,\ell_2 & \mapsto & 
\begin{cases}  
\ell_1+1, \ell_2 &   \text{ if  }  \ell_1 = \ell_2, \text{ and } \ell_1  < 64 \\
64, 64  &  \text{ if  }  \ell_1 = 64 \text{ or } \ell_2  = 64 \\
\ell_1, \ell_2& \text{ otherwise } 
\end{cases}
\end{array}$$
Again, Sensei asks the ninjas to execute the protocol for a certain time, and then attack if they are in level 64.
The protocol can again be visualized as a Petri net, shown on the right of \fref{fig:senseis-protocols}.

Sensei first shows that the protocol works correctly in the absence of the sniper.  Eventually, each level is populated by at most one ninja; indeed, any two ninjas in the same level eventually interact and one of them moves up. Therefore, if there are at least 64 ninjas, then one of them eventually reaches the top of the tower, and eventually brings all other ninjas to the top. 

To prove that the protocol is robust, Sensei observes that the argument above works not only for the initial configuration with all ninjas in the ground level, but \emph{for any configuration with at least 64 ninjas}. Indeed, in any such configuration either one ninja is already at the top, or , by the pigeonhole principle, eventually one ninja moves one level up. Therefore, unless the sniper brings the number of ninjas below 64, the ninjas attack.

\medskip\paragraph{Sensei's question.}
Sensei is not satisfied yet.  At some meetings, the ninjas must decide by majority if they attack or not, that is, they attack only if they are at least 64 and a majority wants to attack. Sensei has a majority protocol for this; in her jargon she says that the protocol \emph{decides} the predicate $Q_2(n_\Yes, n_\No) \Leftrightarrow n_\Yes + n_\No \geq 64  \wedge n_\Yes > n_\No$ (instead of the predicate $Q_1(n) \Leftrightarrow n \geq 64$ of the first  protocol).  She also has protocols for qualified majority (i.e., the ninjas attack only if at least 2/3  of them are in favour, corresponding to the predicate $Q_2(n_\Yes, n_\No)\Leftrightarrow 2 n_\Yes \geq 3 n_\No$), and for whether the number of ninjas showing up is a multiple of 7, corresponding to  $Q_3(n) \Leftrightarrow \exists k \colon n = 7k$.~\footnote{Ninjas are superstitious.} In fact, Sensei has read \cite{AngluinAER07}, in which Angluin et al. introduced this kind of protocols and showed how to construct a protocol  for any predicate expressible in Presburger arithmetic. She has even read the improved version \cite{CzernerGHE24}, which shows how to construct protocols slightly faster than those of \cite{AngluinAER07} and, more importantly, with polynomially instead of exponentially many states in the size of the Presburger formula (assuming it is quantifier-free). However, these protocols are not robust. For example, the protocol of \cite{CzernerGHE24} for the predicate $x \geq 64$ is the first protocol of the introduction. Since these protocols are very dangerous  in the presence of snipers, Sensei has the following question:

\begin{quote}
Is there a robust protocol for every Presburger predicate?
\end{quote}
We do not know yet. In this paper we give a positive partial answer. We construct robust protocols for \emph{threshold predicates} of the form $\sum_{i=1}^k a_i x_i \geq b$, where $a_1, \ldots, a_k, b \in \mathbb{Z}$, and for \emph{modulo predicates} of the form $\sum_{i=1}^k a_i x_i \equiv_c b $, where $a_1, \ldots, a_k \in \mathbb{Z}$, $c \in \mathbb{N}$, and $0 \leq b \leq c-1$. 

\medskip\paragraph{Related work.} The literature on fault-tolarence and robust distributed algorithms is very large, and so we consider only work on population protocols or closely related models. In \cite{Delporte-GalletFGR06}, Delporte-Gallet \emph{et al.} initiated the study of population protocols with essentially the same fault model as ours. The paper provides a construction that, loosely speaking, yields robust protocols  with respect to a given number of failures, independent of the number of agents. The construction we give for modulo protocols is an immediate consequence of this idea. Delporte-Gallet and collaborators have published many other papers on fault-tolerant consensus and leader election algorithm in different kinds of message-passing networks. 

There also exists previous work on other fault models. In \cite{GuerraouiR09}, Guerraoui and Ruppert study community protocols, a model that extends population protocols by assigning unique identities to agents.  They present a universal construction that, given a predicate in NSPACE($\log n)$, outputs a community protocol deciding it. The construction is robust with respect to Byzantine failures of a constant number of agents. More recently,  Alistarh \emph{et al.} have studied robustness of population protocols and chemical reaction networks with respect to leaks \cite{AlistarhDKSU17,Alistarh0U21}. Loosely speaking, this is a failure model in which agents can probabilistically move to adversarially chosen states.

\medskip\paragraph{Structure of the paper.} Section \ref{sec:prelims} introduces the syntax of population protocols and their semantics in the presence of a sniper. Section~\ref{sec:robustness} formalizes the notion of robustness, and formally proves that the first and second of Sensei's protocols are not robust and robust, respectively. Sections~\ref{sec:threshold} and \ref{sec:modulo} give robust population protocols for all threshold and modulo predicates, respectively. Section~\ref{sec:boolean-comb} explains the problem with robustness of the standard construction for Boolean combinations.

\section{Generalized Protocols and Population Protocols with Sniper}
\label{sec:prelims}

\newcommand{\sn}{\hookrightarrow}
\newcommand{\sni}[1]{\hookrightarrow^{#1}}
\newcommand{\tosn}{\Rightarrow}
\newcommand{\tosni}[1]{\Rightarrow^{#1}}
\newcommand{\tostar}{\to^{*}}
\newcommand{\snstar}{\hookrightarrow^{*}}
\newcommand{\tosnstar}{\Rightarrow^{*}}
\newcommand{\tol}{\textit{Tol}}
\newcommand{\intol}{\textit{InTol}}

After some basic notations, we introduce the syntax and the semantics in the presence of a sniper of population protocols. For convenience, we introduce the model as a particular instance of a slightly more general model called generalized protocols.

\medskip\paragraph{Preliminaries: multisets and intervals.} Given a finite set $Q$, a \emph{multiset} is a function $M: Q \rightarrow \N$. We denote the set of multisets over $Q$ by $\N^Q$. Multisets can be added, subtracted and compared componentwise. We define the \emph{size} of a multiset $M$ as $|M| := \sum_{q \in Q} M(q)$. The set of multisets of size $k$ is denoted by $\N^Q_k$. The \emph{support} of a multiset is $\supp{M} := \{q \in Q: M(q) \neq 0\}$. 
We use a set-like notation to write multisets explicitly: $\cfg{x, y, y, 3 \cdot z}$ denotes the multiset $M$ with $M(x) = 1$, $M(y) = 2$, $M(z) = 3$ and otherwise $0$. When the set $Q$ contains natural numbers, we write the numbers from $Q$ in bold to differentiate them from their multiplicity: $\cfg{\mathbf{1}, 2\cdot\mathbf{3}}$.

We write $[a, b)$ to denote the set $\{x \in \mathbb{Z}: a \leq x < b\}$. We define the length of an interval as $\len([a, b)) := b - a$.

%\subsection{Generalized Protocols with Sniper}

\newcommand{\accept}{1}\newcommand{\reject}{0}
\newcommand{\Qpos}{Q_+}
\newcommand{\Qneu}{Q_=}
\newcommand{\Qneg}{Q_-}

\paragraph{Generalized Protocols:  Syntax.} A \emph{generalized protocol} is a tuple: $\Pcal = (Q, \delta, I, O)$	where
	\begin{itemize}
		\item $Q$ is a finite set of states,
		\item $\delta \subseteq \N^Q_{2} \times \N^Q_{2}$ is the transition relation, 
		\item $I \subseteq Q$ is the set of initial states,
		\item $O: \N^Q \rightarrow \{\accept, \reject, \bot\}$ the output function.
	\end{itemize}
    Elements of $\N^Q$ are called \emph{configurations} and elements of $\delta $ \emph{transitions}. Configurations consist of \emph{agents}, which are in a specific state $q\in Q$. We also call agents \emph{ninjas}. Intuitively,  the output function assigns to each configuration either a decision (\accept\ or \reject), or no decision ($\bot$). We let $p, q \mapsto p', q'$ denote that $(\cfg{p,q}, \cfg{p',q'}) \in \delta$. The \emph{preset} and \emph{poset} of a transition $t= (p, q \mapsto p', q')$ are
$\pre(t)   := \cfg{p, q}$ and $\post(t)  := \cfg{p', q'}$. We call a transition $t$ with $\pre(t) = \post(t)$ \emph{silent}.

Additionally we assume that for every pair of states $p, q \in Q$, there exists a transition $t$ with $\pre(t) = \cfg{p, q}$. If the protocol does not define such a transition explicitly, we add the respective identity transition (such a transition is silent and does not change the behavior of the protocol).

\medskip\paragraph{Steps, snipes, and moves.} 
Let $C, D \in \N^Q$ be two configurations.
\begin{itemize}
    \item We write $C \to D$ if $C = D$ or there exists some transition $t$ s.t.\ $\pre(t) \leq C$ and $D = C + \post(t) - \pre(t)$. We call $C \to D$ a \emph{move} of the protocol.

    A configuration $C$ is \emph{terminal} if there is no configuration $D \neq C$ such that $C \rightarrow D$.
\item We write $C \sn D$ if $C = D + S$ for some configuration $S$ such that $|S| = 1$. 
We call $C \sn D$ a  \emph{snipe}.   
\item We write $C \tosn D$ if either $C \to D$ or $C \sn D$, and call $C \tosn D$ a \emph{step}. That is, a step is either a move of the protocol, or a snipe. 
\item We let $\tostar$, $\snstar$, and $\tosnstar$ denote the reflexive and transitive closure of $\to$, $\sn$ and $\tosn$, respectively.
\item If $C \tostar D$ and $E \leq D$, we say that $C$ \emph{covers} $E$.
\end{itemize}

%Given $i \geq 0$ and configurations $C, D$, we say that $D$ is \emph{$i$-reachable} from $C$, denoted $C \tosni{i} D$,  if there is a sequence of steps leading from $C$ to $D$ containing  at most $i$ snipes.

\medskip\paragraph{$i$-executions, $i$-reachability, fairness.}  Fix $i \geq 0$.  Loosely speaking, an $i$-execution is an infinite sequence of steps containing at most $i$ snipes. Formally, an \emph{$i$-execution} is an infinite sequence $\pi = \{C_j\}_{j \in \N}$ of configurations such that $C_j \tosn C_{j+1}$ for every $j \geq 0$, and $C_j \sn C_{j+1}$ for at most $i$ values of $j$. The \emph{output} of $\pi$ is
\[\out(\pi) := \begin{cases}
		\accept    & \text{if } \exists n \in \N: \forall i \geq n: O(\pi(i)) = \accept \\
		\reject   & \text{if } \exists n \in \N: \forall i \geq n: O(\pi(i)) = \reject \\
		\bot & \text{else}
\end{cases}\]
An $i$-execution $\pi$ is \emph{fair} if{}f for every configuration $D \in \N^Q$ the following holds: if $D$ can be reached from infinitely many configurations of $\pi$ by executing sequences of moves, then $D$ appears in $\pi$. It is easy to see that fairness can be equivalently defined as follows: if $D$ can be reached \emph{in one move} from infinitely many configurations of $\pi$, then $D$ appears \emph{infinitely often} in $\pi$. Formally, for every configuration $D$, if  $\pi(j) \rightarrow D$ for infinitely many $j \in \N$, then $\pi(j) = D$ for infinitely many $j \in \N$.

\medskip\paragraph{Predicate decided by a generalized protocol.}
Fix $i \geq 0$. Given a configuration $C \in \N^Q$ we define its $i$-output, denoted $\out_i(C)$, as follows:
\[\out_i(C) := \begin{cases}
		\reject    & \text{if } \forall \pi \text{ fair $i$-execution}, \pi(0) = C: \out(\pi) = \reject \\
		\accept    & \text{if } \forall \pi \text{ fair $i$-execution}, \pi(0) = C: \out(\pi) = \accept \\
		\bot & \text{else}
	\end{cases}\]
%A generalized protocol is \emph{$i$-well-specified} if $\out_i(C) \in \{0,1\}$ for every initial configuration $C \in \N^I$.
A generalized protocol $i$-decides a predicate $\phi \colon \N^I \to \{\accept, \reject\}$  if $\out_i(C)= \phi(C)$ for every $C \in \N^I$. 

\medskip\paragraph{Consensus output functions, population protocols.} 
An output function $O: \N^Q \rightarrow \{\accept, \reject\}$ is a \emph{consensus function}
if there is a partition of $Q$ into two sets $\Qpos$ and $\Qneg$ of \emph{accepting} and \emph{rejecting} states such that for every configuration $C$ we have: $O(C)=\accept$ if$\supp{C} \subseteq \Qpos$, $O(C) = \reject$ if $\supp{C} \subseteq \Qneg$, and  $O(C)=\bot$ otherwise.  A generalized protocol is a \emph{population protocol} if its output function is a consensus function.

Intuitively, if we interpret $\Qpos$  and $\Qneg$ as the set of states in which ninjas have accept and reject \emph{opinions}, respectively, then a configuration of a population protocol has output $b \in \{\accept, \reject\}$ if{}f all ninjas currently have the same opinion, i.e., have currently reached a \emph{consensus}, and $\bot$ otherwise.

%We can now interpret the usual definition of a population protocol as a general protocol:
%
%\begin{definition}
%	A \emph{population protocol} is a 4-tuple: $\Pcal = (Q, \delta, I, O)$ where
%	\begin{itemize}
%		\item $Q$ is a finite set of states,
%		\item $\delta \subseteq Q^2 \times Q^2$ the transition relation,
%		\item $I \subseteq Q$ the set of initial states,
%		\item $O: Q \rightarrow \{0, 1\}$ the output function on the states.
%	\end{itemize}
%	The \emph{underlying general protocol} has the same $Q, I$ and $\delta$. We extend $O$ to
%	configurations by
%	\[ C \mapsto \begin{cases}
%			0    & \text{if } O(\supp{C}) = \{0\} \\
%			1    & \text{if } O(\supp{C}) = \{1\} \\
%			\bot & \text{else}
%		\end{cases}\]
%\end{definition}
\begin{example}
\label{ex:angluin-prot}
We restate the first protocol from the introduction using this formalism for an arbitrary threshold $t$:

\begin{protocol}{Pebbles}
\States{$Q = [0, t]$}
\Input{$I = \{1\}$}
\Output{$O(x) = 1$ if $ x = t$, and $O(x) = 0$ otherwise}
\Transitions{$\delta  = \begin{array}[t]{l}
 \{ (x, y) \mapsto (x + y, 0): x, y \in Q \land x + y < t \}  \; \cup  \\
\{ (x, y) \mapsto (t, t): x, y \in Q \land x + y \geq t\} 
\end{array}$}
\end{protocol}
\end{example}

\section{Robustness}
\label{sec:robustness}

%The following lemma shows that if $D$ is  snipe-reachable from $C$, then it is also reachable by a sequence where snipes occur only at the end.
%\begin{lemma}
%For every $i\geq 0$: $C \tosni{i} D$ if{}f $C \tostar E \sni{i} D$ for some configuration $E$.
%\end{lemma}
%\begin{proof}
%	The if part is trivially true.
%
%	Let $D$ be a configuration that is $i$-snipe-reachable from $C$ via the sequence $C_0, ..., C_n$. Now consider the sequence $\widetilde{C_0}, ..., \widetilde{C_{\tilde{n}}}$ obtained from $(C_j)_j$ by removing all snipings and just keeping the sniped agents in their states indefinitely. We now have that $\widetilde{C_{\tilde{n}}}$ is reachable from $C$ and $D \in \snipe_i(\widetilde{C_{\tilde{n}}})$.
%\end{proof}

Recall the first protocol from the introduction. Even if the initial number of ninjas is 126, the sniper can still prevent the attack  by taking one single ninja down. However,  this is the case only because the sniper can intervene at any moment. Indeed, if the sniper can only intervene before the protocol starts, she must take at least 62 ninjas down to prevent the attack.  We say that the initial configuration with 125 ninjas has an \emph{initial tolerance} of 62, but a \emph{global tolerance} of 1.  It is easy to see that the initial tolerance for an arbitrary configuration with $n \geq 64$ is $n - 64$, while the global tolerance is $\lceil \frac{n}{63} \rceil - 2$.  However, in the second protocol, the initial and global tolerance of any initial configuration with $n \geq 64$ ninjas coincide, and are equal to $n - 64$. 

Observe that the initial tolerance of an initial configuration with $n$ ninjas is $n - 64$ for \emph{any} protocol deciding whether $n \geq 64$ holds. Intuitively, it gives an upper bound on how fault-tolerant a protocol for this predicate can be. On the contrary, the global tolerance depends on the protocol. This suggests to define the class of robust protocols as those protocols whose global tolerance is equal to the initial tolerance of all initial configurations. We proceed to formalize this notion.
\begin{definition}
Let $P$ be a well-specified generalized protocol, and let $C \in \N^I$ be an initial configuration of $P$. 
\begin{itemize}
\item The \emph{initial tolerance} of $C$, denoted $\intol(C)$, is the largest number $i$ such that $\out_0(C) = \out_0(D)$ for every configuration $D$ with $C \sni{i} D$. 
\item The \emph{global tolerance} of $C$, denoted $\tol(C)$, is the largest number $i$ such that $\out_i(C) = \out_0(C)$.
\end{itemize}
$P$ is \emph{robust in} $C$ if $\tol(C) = \intol(C)$. We call $P$ \emph{robust} if it is robust in every initial configuration.
\end{definition}
We note a couple of useful facts about robustness:
\begin{remark}\;
    \label{remark:robust-facts}
\begin{itemize}
    \item Negating a population protocol via changing accepting states to be rejecting states and vice-versa, the protocol still remains robust. This is because neither the initial tolerance nor the global tolerance change.
    \item Sniping agents in initial states that have never interacted cannot change the output when considering snipes less than the initial tolerance. This follows directly from the definition: when sniping a agent in an initial state, it is as if the initial configuration was different to begin with.
\end{itemize}
\end{remark}

\subsection{Examples}
We  show that with our notion of robustness, the original population protocol for threshold
predicates by Angluin et al.~\cite{AngluinADFP06} is not robust.
\begin{proposition}
	The protocol of \eref{ex:angluin-prot} is not robust for any $t \geq 3$.
\end{proposition}
\begin{proof}
    Assume the protocol is robust. Let $q:=t-1\in Q$ denote the last non-accepting state and consider the following fair $1$-execution:
    \[\cfg{(t+1)\cdot \mathbf{1}} \tostar \cfg{\mathbf{q}, \mathbf{2}, (t - 1) \cdot \mathbf{0}} \sn \cfg{\mathbf{2}, (t - 1) \cdot \mathbf{0}}\to ...\]
The first configuration is an initial configuration with an initial tolerance of $1$. Without sniping, it is accepted. However, the last configuration is terminal and rejecting. This contradicts the assumption that the protocol is robust.
\end{proof}
\begin{remark}
We can say even more about the robustness of this protocol. If the sniper is only permitted to snipe a single ninja, then we know that they can at worst snipe a ninja with $t - 1$ pebbles. Thus if the input configuration has $t - 1$ additional ninjas, then the output is still correct. In other words, if the initial tolerance of a configuration is $i$, then the global tolerance of the protocol is $\lfloor \frac{i}{t - 1} \rfloor$.
\end{remark}
In this paper, we focus our attention on robust protocols and as such we do not need a fine-grained analysis of protocols that fail to be robust. But if one wants to analyze non-robust protocols, it is useful to define robustness with a parameter that is the ratio of initial and global tolerance.

We give the formal definition of Sensei's second protocol, which we call the \Tower{} protocol.

\begin{protocol}{Tower}
\States{$Q = [1, t]$}
\Input{$I = \{1\}$}
\Output{$O(x) = \accept \text{ if{}f } x = t$}
\Transitions{$\delta  =  \begin{array}[t]{ll} \{ (x, x) \mapsto (x, x + 1): x \in Q \land x < t\} \; \cup \\
     \{(x, t) \mapsto (t, t): x \in Q,x<t\} \end{array}$}
\end{protocol}
    
\smallskip\noindent Let us prove that the \Tower{} protocol is robust.

\begin{proposition}\label{thm:tower-robust}
The \Tower{} protocol  is robust for every $t$.
\end{proposition}
\begin{proof}
    Let $C \in \mathbb{N}^I$ be an initial configuration. If $|C| < t$, then $\out_0(C) = 0$ and $C$ cannot \emph{cover} the accepting state, i.e.\ it cannot put an agent into $t$. Since the transition relation is monotonic ($C_1\tostar C_2$ implies $C_1+D\tostar C_2+D$ for all $C_1,C_2,D\in\N^Q$), sniping cannot lead to reaching the accepting state. Therefore $\out_i(C) = 0$ for all $i$.

    If $|C| \geq t$, then we know that $|C| \geq t + \intol(C)$ by definition of the initial tolerance.
	To show $\intol(C) = \tol(C)$, we show that the protocol accepts any configuration with at least $t$ agents.

	First, we note that each transition increases the value of $\sum_{q\in Q}q\cdot C(q)$, so only finitely many transitions can be executed. Thus the protocol eventually reaches a terminal configuration $D$ with $|D| \geq t$. If there is an agent in state $t$, then all other agents must be in state $t$ as well — otherwise the second transition could be executed. So $D$ is accepting in this case. If there is no agent in the $t$ state, then we have at least $t$ agents distributed among $t-1$ states, so there is some state $x<t$ with at least two agents. But this is a contradiction to $D$ being terminal, since those two agents can initiate a transition.
\end{proof}

\section{Robust Threshold Protocols}
\label{sec:threshold}

In this section we construct robust protocols for \emph{threshold predicates}, that is, predicates of the form,
\[\phi(x) = \bigg( \sum_{i=1}^n a_i \cdot x_i \geq t \bigg)\]
\noindent  where $a_1,...,a_n, t \in \mathbb{Z}$. We proceed in three steps. \sref{subsec:inhomogeneous} constructs robust protocols for threshold predicates with  positive coefficients $a_i$ and an arbitrary threshold $t$. \sref{subsec:genmajority} does the same for predicates with arbitrary coefficients and threshold $1$; we call them \emph{generalized majority} predicates, because they generalize the majority predicate $x-y \geq 1 \iff x - y > 0$. Finally, \sref{subsec:combination} combines the protocols of the two previous sections to present protocols for arbitrary threshold predicates.

\subsection{Threshold predicates with positive coefficients}
\label{subsec:inhomogeneous}

In this section we study the case $a_i > 0$. We can \Wlog{} assume that \(t \geq a_i\) for all \(i\): If \(t<0\), then \(\varphi\) is equivalent to true, and we are done. Otherwise, since all coefficients \(a_i>0\), we replace every \(a_i\) by \(min(a_i, t)\) to obtain a new predicate \(\varphi'\): Clearly if some ninja for a coefficient \(a_i\) which was decreased occurs, then both \(\varphi\) and \(\varphi'\) are automatically true. Since changing a coefficient which does not occur does not change satisfaction either, \(\varphi\) and \(\varphi'\) are equivalent.

The Black Ninjas have recently recruited a cohort of rookie ninjas. In battle, two veterans are as good as three rookies. When the ninjas meet, the question to decide is whether their battle power is at least 192 (the power of 64 veterans). In other words, the ninjas must conduct a protocol to decide if  $3 v + 2 r \geq 192$, where $v$ and $r$ are the numbers of veterans and rookies, respectively.

The protocol of \cite{AngluinADFP06} for this case is an easy generalization of the protocol for $n \geq 64$. Each ninja brings a pouch, but veterans put three pebbles in it, while rookies only put two. Otherwise the protocol works as before: when ninjas interact, one gives her pebbles to the other, until two interacting ninjas observe that their joint number of pebbles is at least 192. The protocol fails to be robust, for the same reason as the previous one.

Sensei comes up with a generalization of the \Tower{} protocol. She visualizes again a \Tower{} with 192 floors. However, veterans now occupy three consecutive floors, e.g.\ floors 100, 101, and 102, and rookies occupy two consecutive floors, e.g.\ 101 and 102 (see \fref{fig:tower-step}). If two ninjas that interact share at least one floor, then the ninja with the higher lowest floor moves one floor up. For example, if the two ninjas just mentioned interact, then the rookie moves one floor up, occupying now the floors 102 and 103 (if both ninjas have the same lowest floor, then any one of them moves up). The other ninja stays where she is. Initially, all veterans occupy the floors 0, 1, and 2, and all rookies occupy the floors 0 and 1.

Sensei's reasoning is similar to the previous case: eventually each floor is occupied by at most one ninja (although she has to think a bit to establish this). Therefore, a ninja eventually reaches the top floor if{}f $3 v + 2 r \geq 192$. As before, a ninja that reaches the top floor tells any other ninja she interacts with to move to the top floor too. We call this protocol \InhomTower, and proceed to define it formally.

\begin{protocol}{\InhomTower}
The states of the protocol are intervals of natural numbers. For convenience we use semi-closed intervals $[s, e)$. A ninja in state $[s, e)$ occupies the floors $s, s+1, \ldots, e-1$. The initial states are the intervals starting at zero. Formally we set \medskip

\States{$Q = \{[s,s+a_i): i\in [1,n]\wedge0 \leq s\wedge s+a_i \leq t +1 \}$}
\Input{$I = \{[0, a_i): i \in [1, n]\}$}
\Output{$O([s, e]) = \accept \text{ if{}f } e = t + 1$}
\Transitions{}
\Done

\begin{itemize}
\item For every two non-disjoint intervals $[s_1, e_1), [s_2, e_2)  \in Q$ such that $s_1 \leq s_2$ and $e_1, e_2 < t+1$ the second interval moves one step “upwards”:
\[ [s_1, e_1), [s_2, e_2) \mapsto [s_1, e_1), [s_2+1, e_2 +1) \TraName{step}\]
\item For two intervals $[s_1, t + 1), [s_2, e_2) \in Q$ with $e_2\le t$, i.e.\ the first interval has reached the top, the second moves to the top as well.
\[[s_1, t + 1), [s_2, e_2) \mapsto [s_1, t + 1), [t + 1 - (e_2-s_2), t + 1) \TraName{accum}\]
\end{itemize}
\end{protocol}

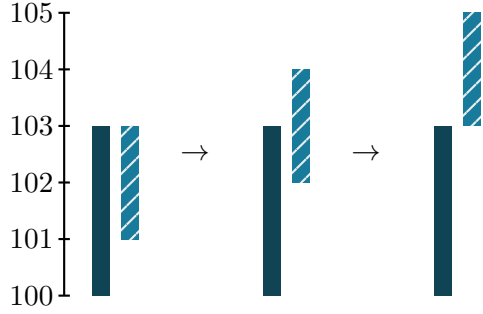
\begin{figure}
    \centering
    \begin{tikzpicture}[thick, scale=0.75]
    	\tikzstyle{agent} = [niceblue];
    	\tikzstyle{agentother} = [preaction={fill, nicebluebright},pattern color=nicebgblue,pattern={Lines[angle=45,distance=5pt]}];
        \draw (-0.5,0) -- (-0.5,5);
        \foreach \i in {0,...,5}
            \node[anchor=east] at (-0.5,\i) {\(10\i\)};
        \foreach \i in {0,...,5}
            \draw[black] (-0.6,\i) -- (-0.4,\i);
        
        \fill[agent] (0,0) -- ++(3mm,0) -- ++(0,3) -- ++(-3mm,0);
        \fill[agentother] (0.5,1) -- ++(3mm,0) -- ++(0,2) -- ++(-3mm,0);

        \node at (1.8,2.5) {$\rightarrow$};

        \fill[agent] (3,0) -- ++(3mm,0) -- ++(0,3) -- ++(-3mm,0);
        \fill[agentother] (3.5,2) -- ++(3mm,0) -- ++(0,2) -- ++(-3mm,0);
        
        \node at (4.8,2.5) {$\rightarrow$};

        \fill[agent] (6,0) -- ++(3mm,0) -- ++(0,3) -- ++(-3mm,0);
        \fill[agentother] (6.5,3) -- ++(3mm,0) -- ++(0,2) -- ++(-3mm,0);
        
    \end{tikzpicture}
    \caption{Two steps using the \TraRef{step} transition between a veteran (dark blue) and a rookie (light blue, striped) in the \InhomTower{} protocol.}
    \label{fig:tower-step}
\end{figure}

\begin{theorem}
	\label{thm:inhom-tower-correct-and-robust}
	\InhomTower{} computes $\phi$ and is robust.
\end{theorem}
\begin{proof}
Given an agent in state $[s,e)\in Q$, we say that it \emph{occupies} levels $s,...,e-1$.

Let $C\in\N^I$ denote an initial configuration and let $k:=\intol(C)$. We show $\out_0(C)=\varphi(C)$ (i.e.\ the protocol has the correct output on $C$) and $\out_k(C)=\out_0(C)$ (i.e.\ the protocol is robust on $C$).

\parag{Case 1: $\sum_{q \in Q} \len(q)C(q)<t$} This implies $\varphi(C)=0$. Here, we claim that no agent reaches a state $[s,t+1)$ and \TraRef{accum} is never executed, which implies $\out_0(C)=0$. To see that this is true, note the invariant that the set of levels in $[0,t]$ occupied by at least one agent is contiguous and contains level $0$. Formally, for any $D$ $0$-reachable from $C$ we have $\bigcup_{D(q)>0}q=[0,i]$ for some $i$. We prove this invariant by induction. The base case follows from the definition of $I$, and for the induction step we use that \TraRef{accum} is never executed, and that in transition \TraRef{step} we have $s_2\in[s_1,e_1)$, so no hole is created, and $s_1\le s_2$, so $0$ remains occupied.

This proves $\out_0(C)=\varphi(C)$. To show $\out_k(C)=\out_0(C)$ we again use a monotonicity argument similar to the proof of \tref{thm:tower-robust}. We have already shown that $C$ cannot cover a state with $t + 1$ as the upper bound without snipes, and this extends to arbitrary $k$-executions.

\parag{Case 2: $\sum_{q\in Q}\len(q)C(q)\ge t$} Note that any transition leaves $\sum_{q \in Q}\len(q)C(q)$ invariant. Let $D$ denote any configuration $k$-reachable from $C$.

If $\sum_{q}\len(q)D(q)<t$, then this observation would imply that it is possible to snipe $k$ agents from $C$ (without executing any moves) and end up in a configuration with output $0$. But this contradicts $k=\intol(C)$.

Again similar to the proof of \tref{thm:tower-robust}, transitions \TraRef{step} and \TraRef{accum} can only be executed finitely often, so we assume \Wlog{} that $D$ is terminal. We now argue that $D$ has output $1$. If level $t$ is occupied, i.e.\ in $D$ there is some agent in a state $[s,t+1)$, then all other agents must be in such a state as well (otherwise \TraRef{accum} could be executed) and the protocol accepts. Otherwise we use $\sum_{q\in Q}\len(q)D(q)\ge t$: an agent in state $q$ occupies $\len(q)$ levels and only levels in $[0,t-1]$ are occupied, so at least one level is occupied by two agents. However, these agents could execute \TraRef{step}, contradicting that $D$ is terminal. 
\end{proof}

\subsection{Threshold predicates with threshold $1$}
\label{subsec:genmajority}

We design a robust protocol for predicates of the form 
\[\phi(x) = \bigg( \sum_{i=1}^n a_i \cdot x_i \geq 1 \bigg)\]
Since they are generalizations of the majority predicate $x - y \geq 1 \iff x - y > 0$, we call them \emph{generalized majority} predicates. In this section, we first introduce weak population protocols. Intuitively, in a weak consensus protocol not all states must be accepting or rejecting, but can can also be neutral. Then we present a very simple and yet robust weak population protocol for generalized majority. Finally, we present a conversion procedure that transforms a weak population protocol into a population protocol for the same predicate that also preserves robustness.

\subsubsection{Weak population protocols for generalized majority}
\label{subsubsec:gods}
Recall that population protocols reach a decision by \emph{consensus}. States can be split into accepting or rejecting. A configuration is a \emph{consensus} if either all ninjas are in accepting states or all ninjas are in rejecting states. The output function assigns a decision to each consensus configuration, and only to them.  We present \emph{weak population protocols} in which states can also be \emph{neutral}. Intuitively, ninjas in neutral states have no opinion. A configuration is a \emph{weak consensus} if no two ninjas have conflicting opinions. The output function assigns a decision to each consensus configuration, and only to them, as follows: accept if at least one ninja wants to accept and no ninjas wants to reject, and reject otherwise. Observe that, in particular, if no ninja wants to accept, then the decision is reject.

\begin{definition}
An output function $O \colon \mathbb{N}^Q \to \{\accept, \reject\}$ is a \emph{weak consensus} function if there is a partition of $Q$ into three sets $\Qpos$, $\Qneu$, $\Qneg$ of \emph{accepting}, \emph{neutral}, and \emph{rejecting} states such that for every configuration $C$:
\begin{itemize}
\item $O(C) = \accept$ if $\supp{C} \subseteq \Qpos \cup \Qneu$ and $\supp{C} \cap \Qpos \neq \emptyset$;
\item $O(C) = \reject$ if $\supp{C} \subseteq \Qneg \cup \Qneu$; and
\item $O(C) = \bot$ otherwise (that is, $O(C) = \bot$ if $\supp{C} \cap \Qpos \neq \emptyset \neq \supp{C} \cap \Qneg$).
\end{itemize}
A generalized protocol $\Pcal = (Q, \delta, I, O)$ is a \emph{weak population protocol} if $O$ is a weak consensus function, and every \(\accept\)-consensus is stable.
\end{definition}
Notice that the output function is determined by the partition $\Qpos$, $\Qneu$, $\Qneg$.
Let us now give weak population protocols for generalized majority predicates. We consider the predicate $x_1 - 2x_2 \geq 1$, which is already representative. It corresponds to the following situation. Sensei wants that the ninjas conduct a vote to decide whether they want to attack or not. The ninjas will only attack if more than 2/3 of them vote for it. Letting $x_1$ and $x_2$ be the number of ninjas in favor of and against attacking, respectively, the ninjas will attack if $x_1 / (x_1 + x_2) > 2/3$ or, equivalently, $x_1 - 2x_2 \geq 1$. Intuitively, a negative vote has the same effect as two positive votes.

Sensei asks the ninjas to bring a pouch containing one \emph{positive pebble} if they want to attack, and two \emph{negative pebbles} otherwise. If two ninjas with pebbles of the same kind interact, nothing happens. If the pebbles are of opposite kinds, one ninja gives her pebbles to the other. However, when a positive and a negative pebble touch each other, they magically disappear, and so after an interaction between two ninjas with $a$ positive and $b$ negative pebbles one ninja has $|a-b|$ pebbles and the other $0$. We visualize this protocol in \fref{fig:gen-maj}.

The possible states of the protocol are $\{1, 0, -1, -2\}$, which gives the number of pebbles in the pouch and their kind. The partition of states is the natural one: $\Qpos = \{1\}$, $\Qneu = \{0\}$, and $\Qneg = \{-1, -2\}$.  The only non-silent transitions are $1, -2 \mapsto 0, -1$, and $1, -1 \mapsto 0, 0$. Observe that, in particular, $1,1 \mapsto 2, 0$ or $-2, -2 \mapsto -4$ are \emph{not} transitions, because the pebbles of the interacting ninjas are of the same kind. 

Let us informally argue that the protocol is correct and robust. Consider a configuration $C$ in which $x_1$ ninjas want to attack, and $x_2$ ninjas do not.  Observe first that the total number of pebbles never increases, and it decreases whenever two ninjas carrying pebbles of opposite kinds interact. Therefore, starting from $C$ the protocol eventually reaches a terminal configuration $C'$ in which all remaining pebbles (if any) are positive or all are negative. In the first case only states of $\Qpos \cup \Qneu$ are populated, and in the second only states of $\Qneg \cup \Qneu$. So $C'$ is a weak consensus. Further, it is easy to see that the cases occur when  $x_1 - 2x_2 \geq 1$ and $x_1 - 2x_2 < 1$, respectively. To prove robustness, observe that the argument above holds for any configuration $C$, and so for the sniper all configurations are equally good.

Let us now give the precise definition of the protocol for the predicate $\sum_{i=1}^n a_i \cdot x_i \geq 1$.
Let $a_\mathit{min},a_\mathit{max}$ be the minimum and maximum of $\{a_1, \ldots, a_n\}$. We assume $a_\mathit{min} < 0 < a_\mathit{max}$, the case of  predicates in which all coefficients have the same sign was already considered in the previous section.  The protocol is:

\begin{protocol}{GenMajority}
\States{$Q = [a_\mathit{min}, a_\mathit{max}]$}
\Input{$I = \{a_1, \ldots, a_n\}$}
\Output{Given by $\Qpos = [1, a_\mathit{max}]$, $\Qneu = \{0\}$, and $\Qneg = [a_\mathit{min},-1]$.}
\Transitions{$\delta = \{x, y \mapsto x + y, 0 : x, y \in Q \land x < 0 < y\}$}
\end{protocol}

\begin{figure}
	\centering
	\begin{tikzpicture}[thick, transform shape, scale=0.75]
        \node[place,nicered,fill=nicered!20!white,label={0:{\Large $\mathbf{-2}$}}] (qm2) {};
		\node[transition] (t1) [below = of qm2] {};
        \node[place,niceredbright,fill=niceredbright!20!white,label={0:{\Large $\mathbf{-1}$}}] (qm1) [below = of t1] {};
        \node[place,nicebluebright,fill=nicebluebright!20!white,label={0:\textbf{{\Large 1}}}] (qp1) [right = of qm1] {};
        \node[place,gray,fill=gray!20!white,label={0:\textbf{{\Large 0}}}] (q00) [left = of qm1] {};
		\node[transition] (t3) [below = of qm1] {};

		\path[->]
		(qm2) edge (t1)
		;
		\path[->]
		(qp1) edge (t1)
		(qp1) edge (t3)
		;
		\path[->]
		(qm1) edge (t3)
		;
		\draw[->](t1) edge (qm1);
		\path[->]
		(t1) edge (q00)
		(t3) edge node[below left] {$2$} (q00)
		;
	\end{tikzpicture}
	\caption{Weak population protocol for $x_1 - 2x_2 \geq 1$. Red states are rejecting, the blue state is accepting and the gray state is neutral.}
	\label{fig:gen-maj}
\end{figure}
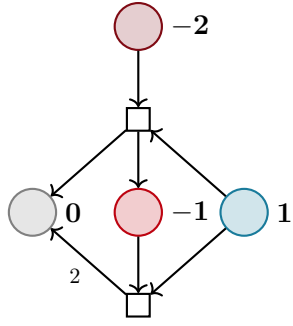

\subsubsection{From weak population to population protocols}
In order to use weak population protocols in practice, we provide an algorithm that given a weak population protocol deciding a predicate, under the assumption that the protocol is \emph{terminating}, i.e.\ every fair execution terminates, outputs a population protocol for the same predicate. Further, the algorithm preserves robustness.

The construction to obtain a population protocol \(\Pcal'\) from a weak population protocol \(\Pcal\) is rather simple: Every ninja remembers a state of \(\Pcal\), and in fact runs the protocol \(\Pcal\) as normal. In addition, neutral ninjas \(A\), i.e.\ in a state \(q \in \Qneu\), get an output bit \(\{+,-\}\), which they set as follows: Whenever they meet a ninja \(B\) with an active opinion, i.e.\ with state \(p \in \Qpos \cup \Qneg\), \(A\) updates her bit to the opinion of \(B\). Since by definition of weak consensus eventually only one opinion remains, every neutral ninja eventually gets the opinion of the remaining ninja. A slight difficulty, where we need the assumption of termination of \(\Pcal\), is the case of only neutral ninjas remaining: We add a transition which takes neutral ninjas with output bits set to \(+\) and \(-\) respectively and sets both bits to \(-\).

\begin{construction}
	\label{cstr:ranks}
	Given a weak population protocol $\Pcal = (Q, \delta, I, O)$ with partition $\Qpos$, $\Qneu,\Qneg$, we construct a population protocol $\Pcal' = (Q', \delta', I', O')$:

    \begin{protocol}{WeakConvert}
    \States{$Q' = \Qpos \cup \Qneg \cup \Qneu \times \{+, -\}$}
    \Input{$I' = \{q : q \in I \cap (\Qpos \cup \Qneg)\} \cup \{(q, -) : q \in I \cap \Qneu\}$}
    \Output{$O'(q') = \begin{cases}
    1 & \text{if } q' \in \Qpos \text{ or } q' = (q, +) \text{ with } q \in \Qneu \\
    0 & \text{if } q' \in \Qneg \text{ or } q' = (q, -) \text{ with } q \in \Qneu
    \end{cases}$}
    \Transitions{}
    \Done
    \begin{itemize}
        \item First we need a transition to simulate $\Pcal$, given agents in states $p,q$ in $\mathcal{P}'$, they 1) disregard their sign if they are in \(\Qneu\), 2) perform \(\delta\) on the so obtained projected states \(p_{pr}, q_{pr}\), and 3) readd a \emph{negative sign} if they end in \(\Qneu\). 
        
        Formally, let \(\pr \colon Q' \to Q, q' \mapsto q'\) if \(q' \in \Qpos \cup \Qneg\) and \(q' \mapsto q\) if \(q'=(q,s) \in \Qneu \times \{+,-\}\). Also let \(\inj \colon Q \to Q'\), \(q \mapsto q\) if \(q \in \Qpos \cup \Qneg\) and \(q \mapsto (q,-)\) otherwise. For every non-silent transition $(p_{pr}, q_{pr} \mapsto p_{pr}', q_{pr}') \in \delta_\Gcal$, we add for all states \(p,q\) with \(\pr(p)=p_{pr}, \pr(q)=q_{pr}\), the following transition:
	\[p,q \mapsto \inj(p_{pr}'), \inj(q_{pr}') \TraName{derived}\]
        \item Additionally, for $q_+ \in \Qpos$ and $q_= \in \Qneu$ we add the following transition:
            \[q_+, (q_=, -) \mapsto q_+, (q_=, +) \TraName{witnessPos}\]
        \item Dually, for $q_- \in \Qneg$ and $q_= \in \Qneu$ we add the following transition:
            \[q_-, (q_=, +) \mapsto q_-, (q_=, -) \TraName{witnessNeg}\]
        \item And for every $q, q' \in \Qneu$ we add:
	\[(q, -), (q', +) \mapsto (q, -), (q', -)\TraName{convince}\]
    \end{itemize}
    \end{protocol}

	Given an execution $\pi$ on $\Pcal'$, we define the induced execution $\tilde{\pi}$ on
	$\Pcal$ via $\tilde{\pi}(i) = \text{pr}(\pi(i))$ where $\text{pr}((q, s)) = q$.
	Moving from \(\pi\) to $\tilde{\pi}$ intuitively replaces \TraRef{derived} by their respective transition and
    \TraRef{witnessPos}/\TraRef{witnessNeg}/\TraRef{convince} by silent identity transitions.
	If $\pi$ is fair, $\tilde{\pi}$ is also fair.
\end{construction}

In the following, a positive voter is a ninja with a state in \(\Qpos\), a negative voter a ninja with a state in \(\Qneg\) and an active voter is either a positive or negative voter. Essentially there is only one type of configuration of \(\Pcal'\) where assigning the output as above will fail: If there are no active voters, and neither is a ninja in \(\Qneu \times \{-\}\) to convince the others. Hence we call a configuration \(C\) \emph{dangerous} if \(\supp{C} \subseteq \Qneu \times \{+\}\), and \emph{harmless} otherwise. 

\begin{lemma}
	\label{lma:god-output}
	Let \(C_0'\) be a harmless configuration of \(\Pcal'\) such that all fair executions starting at
	$\pr(C_0') \in \N^{Q}$ terminate with output $b$ in $\Pcal$. 
	
	Then all fair executions starting at $C_0'$ terminate with output $b$ in $\Pcal'$.
\end{lemma}
\begin{proof}
	Let $\pi=(\pi_0, \pi_1, \dots)$ be a fair execution in $\Pcal'$ starting at $C_0'$. We have to show that \(\pi\) has output \(b\). 
	First consider the induced execution $\tilde{\pi}$ in $\Pcal$. Since $\Pcal$ is terminating and \(\tilde{\pi}\) is fair, \(\tilde{\pi}\) terminates 
	in some configuration \(C_{pr}'\). Clearly using \TraRef{witnessPos}/\TraRef{witnessNeg} and potentially \TraRef{convince}, 
	\(\pi\) then also terminates in some configuration \(C'\), and we know that \(\pr(C')=C_{pr}'\).
	Because $\tilde{\pi}$ is also fair, the output of $C_{pr}'$ is \(b\). If an active voter exists in \(C_{pr}'\) and hence in \(C'\), it will eventually 
	convert everyone to the correct output via the \TraRef{witnessPos}/\TraRef{witnessNeg} transitions, and hence \(\pi\) again has output \(b\).

	If \(C'\) does not contain an active voter, then \(C_{pr}'\) has output \(0\) by definition, and we have to show the same for the execution \(\pi\).
	We first show that the configuration \(C'\) contains an agent in \(\Qneu \times \{-\}\). We prove this via case distinction:
	
	\parag{Case 1: \textnormal{No \TraRef{derived} transition ever occurred}} Then we have \(\pi_i(\Qpos)=0\) for all \(i\), and since \(C_0'\) is harmless, initially an agent was in \(\Qneu \times \{-\}\). The only transition which could remove the agent from \(\Qneu \times \{-\}\), namely \TraRef{witnessPos}, was never enabled.
	
	\parag{Case 2: \textnormal{\TraRef{derived} did occur}} Then consider the configuration \(\pi_i\) directly after the last occurrence. It contains an agent in \(\Qneu \times \{-\}\) by definition of \TraRef{derived}. We can now argue as in case 1.
	
	It follows that \TraRef{convince} moves all agents from \(\Qneu \times \{+\}\) to \(\Qneu \times \{-\}\), and \(\pi\) terminates with output \(0\).
\end{proof}
\begin{theorem}
	\label{thm:god-robust}
	Given a terminating weak population protocol $\Pcal$ deciding a predicate \(\varphi\), the above construction outputs a population protocol $\Pcal'$ deciding \(\varphi\). Additionally, if $\Pcal$ is robust, then $\Pcal'$ also is.
\end{theorem}
\begin{proof}
	%We use \csref{cstr:ranks} to obtain a population protocol $\Pcal'$.
	Applying \lref{lma:god-output} to all initial configurations, we conclude that $\Pcal'$
	decides the same predicate \(\varphi\) as $\Pcal$. Observe that initial configurations are harmless since agents start in \(\Qneu \times \{-\}\) by definition.

    It remains to prove that \(\Pcal'\) is robust. Hence let $C_0' \in \N^{I'}$ be an initial configuration. We have to show that \(\tol(C_0') = \intol(C_0')\). Hence let \(\pi\) be a fair $i$-execution starting at \(C_0'\) with \(i = \intol(C_0')\). We have to show that \(\pi\) has output \(\varphi(C_0')\). 
    
    Observe that since \(\Pcal\) is terminating, \(\Pcal'\) is also terminating. Hence \(\pi'\) eventually becomes constant, we call the reached terminal configuration \(C''\). Even if an agent was sniped earlier in \(\pi\), we can \Wlog{} assume it was sniped directly before \(C''\), and simply did not participate in transitions before that. I.e. there exists $C' \in \N^{Q'}$ such that $C_0' \tostar C'$ and $C' \sn^i C''$. To show that \(\pi\) has output \(\varphi(C_0')\), it suffices to show that \(C''\) has output \(\varphi(C_0')\). We want to apply \lref{lma:god-output} on the configuration \(C''\), hence we have to check its assumptions.
    
    First, since \(\Pcal\) and \(\Pcal'\) compute the same predicate, we have \(\intol_{\Pcal'}(C_0')=\intol_{\Pcal}(\pr(C_0'))\). Since \(\Pcal\) is robust, and the projected execution \(\tilde{\pi}\) contains at most \(i = \intol_{\Pcal}(\pr(C_0'))\) snipes, \(\tilde{\pi}\) has output \(\varphi(C_0')\). Therefore \(\pr(C'')\) has output \(\varphi(C_0')\). Since \(\pr(C'')\) is terminal, this shows the assumption about all fair executions of the projection having correct output. It remains to prove that \(C''\) is harmless.

    Assume for contradiction that \(C''\) is dangerous. In particular \(\pr(C'')\) is rejected. Similar to the case reasoning in \lref{lma:god-output}, observe that the only transition outputting agents in \(\Qneu \times \{+\}\) is \TraRef{witnessPos}. Hence to reach a dangerous configuration, in some configuration between $C_0$ and $C'$ some agent must have been in \(\Qpos\) and got sniped. We denote the latest (i.e.\ closest to $C'$) configuration that has a positive voter by $C_+$.

    Note that no agent in $C''$ participated in a transition since $C_+$, since after $C_+$ no positive voters exist and hence all enabled transitions set the output bit to \(-\). We now consider the configuration $C_+'' = C'' + \cfg{q_+}$ where $q_+$ is a positive voter in $C_+$. We can reach $C_+''$ from $C_+$ via sniping exactly those agents that are sniped from $C'$ to $C''$ except for the positive voter. Since \(C''\) does not contain negative voters, \(C_+''\) does not, and hence \(C_+''\) is an accepting consensus. By assumption on weak population protocols, any accepting consensus is stable, hence \(C_+''\) is accepted. Hence \(\pr(C_0')\) has an accepting execution via \(C_+''\), and a rejecting execution via \(C''\), despite both executions sniping \(\leq \intol(\pr(C_0'))\) many agents. Contradiction to robustness of \(\Pcal\).
\end{proof}

\begin{example}
We now apply this construction to the protocol that we gave above for $x_1 - 2 x_2 \geq 1$ and visualized in \fref{fig:gen-maj}.
We write $+0$ and $-0$ instead of $(0,+)$ and $(0,-)$ for clarity.
\begin{protocol}{SignedNumbers}
\States{$Q = \{1, -1, -2, +0, -0\}$}
\Input{$I = \{-2, 1\}$}
\Output{$O(q) = \begin{cases}
    1 & \text{if } q = 1 \text{ or } q = +0 \\
    0 & \text{if } q = -1 \text{ or } q = -2 \text{ or } q = -0
\end{cases}$}
\Transitions{}
\Done
\[\begin{array}[t]{rrcrrr}
        -2    ,& 1   & \; \mapsto \; & -1   ,&+0 \qquad & \text{coming from \TraRef{derived}}\\
        -1    ,& 1      & \mapsto & -0,& -0 \qquad & \text{coming from \TraRef{derived}}\\
        -0,& 1      & \mapsto & +0,& 1      \qquad & \text{coming from \TraRef{witnessPos}}\\
        +0,& -2     & \mapsto & -0,& -2     \qquad & \text{coming from \TraRef{witnessNeg}}\\
        +0,& -1     & \mapsto & -0,& -1     \qquad & \text{coming from \TraRef{witnessNeg}}\\
       +0,& -0 & \mapsto & -0,&-0 \qquad & \text{coming from \TraRef{convince}}
    \end{array}\]
\end{protocol}
    
\end{example}

\subsection{Arbitrary threshold predicates}
\label{subsec:combination}

In this subsection we again consider the threshold predicate from the beginning of this section:
\[\phi(x) = \bigg( \sum_{i=1}^n a_i \cdot x_i \ge t \bigg)\]
Now we also allow for $a_i, t \in \mathbb{Z}$. We again will create a tower-esque protocol, however
this time it will be a weak population protocol. Ninjas representing inputs with negative coefficients will have to cancel positive ninjas. We assume \(t \geq 1\) \Wlog, since we can rewrite \(\sum_{i=1}^n a_i \cdot x_i \ge t\) as \(\sum_{i=1}^n (-a_i) \cdot x_i \le -t\), which is the negation of \(\sum_{i=1}^n (-a_i) \cdot x_i \ge -t+1\). Negation preserves robustness, see \rref{remark:robust-facts}.

\newcommand{\Qt}{Q_{\textsc{tower}}}
\newcommand{\Qc}{Q_{\textsc{cancel}}}

\begin{protocol}{InhomTowerCancel}
States are intervals of integers as well as non-positive integers. We identify $0$ with all empty intervals. Initial states are intervals starting at zero as well as negative integers.\smallskip

\States*{
    $Q = \Qt \cup \Qc \cup \{0\}$ where
    \item $\Qt = \{[s, e) : 0 \leq s < e \leq T\}$
    \item $\Qc = \{x : \min_i a_i \leq x < 0\}$
    \item $T = \max \{t, a_1, \ldots, a_n\}$
}
\Input{$I = \{[0, a_i) : a_i > 0\} \cup \{a_i : a_i < 0\}$}
\Output*{
    Given by the following partition:
    \item $\Qpos = \{[s, e) \in \Qt : t \leq e\}$,
    \item $\Qneu = \{[s, e) \in \Qt : e < t\} \cup \{0\}$,
    \item $\Qneg = \Qc$
}
\Transitions{}
\Done
\begin{itemize}
    \item For every two non-disjoint intervals $[s_1, e_1), [s_2, e_2) \in \Qt$ with $e_1, e_2 < T$ and $s_1 \leq s_2$ the transition
    \[[s_1, e_1), [s_2, e_2) \mapsto [s_1, e_1), [s_2 + 1, e_2+1) \TraName[cancel]{step}\]
    \item When a negative agent meets an agent above level $t$, they both reduce their absolute value by one. For all $0 \leq s < e \leq T$ with $t \leq e$ and $x \in \Qc$ the transition
    \[[s, e), x \mapsto [s, e - 1), x + 1 \TraName{cancel}\]
    note that we identify empty intervals i.e.\ $[s, s)$ with the state $0$.
\end{itemize}
\end{protocol}

\newcommand{\Inv}{\operatorname{\mathit{Sum}}}
%\begin{lemma}
%	\label{lma:annihi-req}
%	Let $C \in \N^Q$ be reachable from an initial configuration. Then \TraRef{cancel} can
%	only be executed if $\sum_{q \in Q_T} \len(q) \cdot C(q) > t$ and there exists some agent in a
%	state from $Q_C \setminus \{0\}$.
%\end{lemma}
%\begin{proof}
%	In order for \TraRef{cancel} to be executed, an agent must be at the top of the tower. This
%	could have only happened via execution of the \TraRef[cancel]{step} transition. But this transition only
%	moves one of the two agents upwards. By repeating this argument, we obtain for each slot $[x, x + 1)$
%	an agent that occupies it. Since the total height of the tower is $t + 1$, we conclude that the
%	sum of the lengths of all agents must be greater than $t$.
%\end{proof}

\begin{theorem}
	\label{thm:inhom-tower-cancel-correct-and-robust}
	\InhomTowerCancel{} computes $\phi$ and is robust.
\end{theorem}
\begin{proof}
    This proof is similar to that of \tref{thm:inhom-tower-correct-and-robust}.
    Given an agent in state $[s, e) \in \Qt$, we say that it \emph{occupies} levels $s, ..., e - 1$.

    Let $C_0 \in \N^I$ be an initial configuration and let $k = \intol(C_0)$. Additionally, let $\Inv(C) := \sum_{q\in \Qt} \len(q) C(q) + \sum_{q\in \Qc} q C(q)$ for any configuration $C$. Observe that any transition leaves $\Inv$ invariant. Since $\Inv(C) \leq t$ is $\phi$, initial tolerance ensures that sniping does not change $\Inv(C) \leq t$.

    Additionally, we note that the protocol produces an output for every fair $n$-execution for arbitrary $n$. Since eventually no agent reaches a level above or equal to $t$, or no negative states contain agents.
    We now show $\out_0(C_0) = \phi(C_0)$ and $\out_k(C_0) = \out_0(C_0)$.

    \parag{Case 1: $\Inv(C_0) < t$} This implies $\phi(C_0) = 0$. For contradiction, we consider a fair $i$-execution starting at $C_0$ that is accepting with minimal $i \leq k$. Since the protocol is bounded, we can equivalently consider configurations $C_0 \tostar C \snstar D \tostar C_{\mathit{term}}$ where $C_{\mathit{term}}$ is terminal.

    We now show that we can assume that the snipes from $C$ to $D$ do not snipe agents that are in $\Qt$. Let $\Delta$ denote the agents in $\Qt$ that were sniped in $C \snstar D$. Now consider $C' = D + \Delta$: we have $C \snstar C' \tostar C_{\mathit{term}} + \Delta$. The last configuration is again accepting, since it does not contain any agents in $\Qc$. While it is not terminal, it cannot become rejecting, since agents in $\Qc$ can never again exist. If $|\Delta| > 0$, then we have a smaller counterexample.

    In the same manner as in the proof of \tref{thm:inhom-tower-correct-and-robust}, we show that the set of occupied levels in a configuration $D$ $k$-reachable from $C$ without sniping in $\Qt$ is contiguous and contains $0$. We show by induction that $\bigcup_{D(q) > 0} q = [0, i]$ for some $i$. The base case follows from the definition of $I$, in the step case we must consider the two transitions: \TraRef[cancel]{step} does not introduce holes, since $s_2 \in [s_1, e_1)$ and $s_1 \leq s_2$, so $0$ remains occupied. \TraRef{cancel} either reduces the upper bound of an interval at or above $t$ by one, this also does not create a hole and leaves $0$ occupied.

    Now we consider the output of $C_{\mathit{term}}$. If there are agents in $\Qc$, then the levels $[t, T)$ cannot be occupied, since otherwise \TraRef{cancel} could be executed. If no agents are in $\Qc$, then we also know that $[t, T)$ cannot be occupied, since $\Inv(C_{\mathit{term}}) = \sum_{q\in \Qt}\len(q)C_{\mathit{term}}(q) < t$ and the occupied levels are contiguous and include 0.

    \parag{Case 2: $\Inv(C_0) \geq t$} This implies $\phi(C_0) = 1$. Let $D$ be a configuration $k$-reachable from $C_0$.

    If $\Inv(D) < t$, then by the observation above it would be possible to snipe $k$ agents in $C_0$ and obtain an initial configuration with $\Inv < t$. By the previous case, it would reject, thus contradicting that $k = \intol(C_0)$.

    We again, similar to the proofs \tref{thm:tower-robust} and \tref{thm:inhom-tower-correct-and-robust} use the fact that \TraRef[cancel]{step} and \TraRef{cancel} can only be executed finitely often. So \Wlog{} $D$ is terminal. We show that $D$ has output $1$: if level $t$ or any above is occupied, then since $D$ is terminal, there cannot exist any agent in $D$ that is in a state from $\Qc$. Therefore $D$ has output $1$. In the case that no agent occupies level $t$ or above, we know from $\Inv(D) \geq t$ that $\sum_{q\in \Qt}\len(q)D(q) \geq t$ and thus at least one level is occupied by two agents. This means that \TraRef[cancel]{step} could be executed contradicting that $D$ is terminal.
\end{proof}
\begin{corollary}
	Giving \InhomTowerCancel{} as input to \csref{cstr:ranks} we obtain a robust population protocol computing $\phi$.
\end{corollary}
\begin{proof}
	Clearly \InhomTowerCancel{} is terminating, and \tref{thm:inhom-tower-cancel-correct-and-robust} shows it is robust. Every accepting consensus is stable: Namely no agent is in \(\Qc\) since these are negative voters, wherefore \TraRef{cancel} is permanently disabled. On the other hand \TraRef[cancel]{step} can only cause additional agents to become positive voters. Hence we can apply \tref{thm:god-robust}.
\end{proof}

\section{Robust Modulo Protocols}
\label{sec:modulo}

Throughout this section, fix a modulo predicate not equivalent to true or false:
\[\phi(x) = \bigg(\sum_{i=1}^n a_i \cdot x_i \bmod m \geq t\bigg)\]
With $a_i, t \in \mathbb{Z}$ and $m \in \N$. Without loss of generality, we can assume that $0 < a_i, t < m$, because the sum is calculated modulo $m$.

When the Black Ninjas attack fortresses of the Dark Powers on nights with a full moon, they only do so if the number of ninjas modulo 7 is equal to \(4,5\) or \(6\). Therefore, Sensei is searching for a robust population protocol for $x \bmod 7 \geq 4$. Sensei finds an important difference between threshold and modulo predicates: Threshold predicates allow for arbitrarily large initial tolerance, for modulo predicates however there exists an upper bound on the initial tolerance.
\begin{proposition}
	\label{prop:modulo-pred-safe-empty}
    For any population protocol computing $\phi$, we have for any initial configuration $C \in \N^I$ that $\intol(C) < m$.
\end{proposition}
\begin{proof}
    Using at most $m - 1$ snipes, the value of the sum can be changed by at least $m - 1$, thus allowing to change the output.
\end{proof}
Using this fact, we will first construct a protocol for sufficiently large inputs in \sref{subsec:big-modulo}.
Then we will combine that protocol in \sref{subsec:small-modulo} with an \InhomTower and a protocol for small inputs to obtain a robust population protocol for all inputs.

\subsection{A Modulo Protocol for Big Inputs}
\label{subsec:big-modulo}

After proving \pref{prop:modulo-pred-safe-empty}, Sensei comes up with a robust population protocol that can handle sufficiently large initial configurations.
The sniper snipes at most $m-1$ ninjas, allowing Sensei to take advantage of a principle common with error-correcting codes: to achieve a fault tolerance of $r$, $2r + 1$ copies of a fault intolerant code are used in parallel. Then one can decide the correct value via majority. Sensei applies this principle to protocols: she thinks that it will be sufficient to have $2(m - 1) + 1 = 2m - 1$ copies of any protocol for a modulo predicate running in parallel. For combining with small inputs, which will become clear later, we need an additional copy, so we will use $2m$ copies in total.

To compute $x \bmod 7 \geq 4$, Sensei gives the following instructions to the ninjas:
Each ninja must bring 14 differently colored pouches, each containing a magic pebble. Whenever 7 magic pebbles are in the same pouch, they annihilate each other. Thus the pebbles are added modulo 7.
In the beginning, each ninja must choose a color. After doing so, they become a \emph{leader} for that color.
When a leader and any other ninja that is not a leader for the same color interact, the leader steals all pebbles from the other's pouch with the chosen color.
When two leaders for the same color interact, one of them retires and gives all the pebbles of that color to the other.
The retired leader then again chooses a color where she still has some pebbles left and can again become a different leader. If she has no pebbles left, she remains a non-leader.
Additionally, every ninja remembers for every color the number of pebbles in the pouch of the last seen leader of that color.
The ninjas decide to attack via majority: if a majority of leaders agree that there at least 4 pebbles in their pouch, then they attack.

Sensei argues that the protocol works for sufficiently large configurations (where $x \geq 14$): every ninja can at most be a leader for a given color once. At least one leader for every color will remain and since there are more than 14 ninjas, every color will have a leader. Every leader correctly computes the value $x \bmod 7$ via the magic pebbles and without a sniper, they all will agree. In the case of a sniper, for every snipe at most one color can be disturbed. Since a sniper can at most snipe 6 times, the majority of 14 leaders will remain correct.

We now give the formal definition of the $\BigModulo$ population protocol, where we write \(\one{A}\) for the indicator function of the set \(A\):

\begin{protocol}{\BigModulo}
\States{$Q = [0, 2m] \times \mathbb{Z}_m^{2m} \times \{0, 1\}^{2m}$}
\Input{$I = \{(0, a_i \cdot \one{[1, 2m]}, \zero): 1 \leq i \leq n\}$}
\Output{$O(\_, \_, r) = 1 \iff |\{j: r_j = 1\}|>m$}
\Transitions{}
\Done
\begin{itemize}
    \item We need to distribute the agents in the different copies. For this, we use non-determinism, for all $(0, v, r), q \in Q$ and $1 \leq i \leq 2m$ with $1 \leq v(i)$ we add the following transition:
\[(0, v, r), q \mapsto (i, v, r), q \TraName{distrib}\]
\item When two agents meet that are in different copies, both steal the tokens from their respective copy from the other. For all $(i, v, r), (j, w, s) \in Q$ with $1 \leq i,j \leq 2m$ and $i\neq j$ we add the following transition:
\[(i, v, r), (j, w, s) \mapsto (i, v', r), (j, w', s) \TraName{steal}\]
Where
\[v' = v - v_j \cdot \one{\{j\}} + w_i \cdot \one{\{i\}}\qquad
	w' = w - w_i \cdot \one{\{i\}} + v_j \cdot \one{\{j\}}\]
\item When two agents from the same copy meet, one of them takes all tokens from the other and the other retires. For all $(i, v, r), (i, w, s) \in Q$ with $1 \leq i \leq 2m$ we add the following transition:
\[(i, v, r), (i, w, s) \mapsto (i, v + w_i \cdot \one{\{i\}}, r), (0, w - w_i \cdot \one{\{i\}}, s)\TraName{retire}\]
\item Lastly we need a transition that updates the last component to keep track of the overall output of
the protocol. For all $(i, v, r), (j, w, s) \in Q$ with $1 \leq i \leq 2m$ and
$0 \leq j \leq 2m$ we add the following transition:
\[(i, v, r), (j, w, s) \mapsto (i, v, r'), (j, w, s') \TraName{result}\\\]
Where \[r'(l) = \begin{cases}
		v_i \geq t & \text{if } l = i \\
		r(l)       & \text{else}
	\end{cases} \qquad s'(l) = \begin{cases}
		v_i \geq t & \text{if } l = i \\
		s(l)       & \text{else}
	\end{cases}\]
\end{itemize}
\end{protocol}

We call agents with a non-zero first component \emph{leaders}, and call the \(j\)-th copy of \(\Z^m\) the \(j\)-th subprotocol.
\begin{proposition}
    \label{prop:big-modulo-correct}
	For every initial configuration $C_0 \in \N^I$ with $|C_0| \geq 2m$, $\BigModulo$ decides $\phi$.
\end{proposition}

\begin{proof}
Let \(C_0\) be an initial configuration. We first claim that every subprotocol will eventually have exactly one associated leader.  

Proof of claim: When two leaders meet via \TraRef{retire}, only one of them remains a leader and the other gives away their tokens in that subprotocol. Since \TraRef{distrib} only makes agents leaders that have tokens in the respective subprotocol, that agent will never again be a leader for that subprotocol. Since we have at least $2m$ agents, every subprotocol will eventually end up with a leader.

Now it is easy to see that each leader will eventually hold all the tokens in their respective subprotocol. Namely they accumulate them via \TraRef{steal}.

Since the \(j\)-th sum $\Inv_j(C) := \sum_{q=(i,v,r)\in Q} v(j) \cdot C(q)$ is invariant modulo \(m\), the leader for every subprotocol has exactly $\Inv_j(C_0) \bmod m$ many tokens, hence obtaining the correct output. This output is then broadcasted via the \TraRef{result} transition to every non-leader.
\end{proof}

\begin{proposition}
For every initial configuration \(C_0 \in \N^I\) with \(|C_0| \geq 3m\), $\BigModulo$ is robust.
 \label{prop:big-modulo-robust}
\end{proposition}
\begin{proof}
Let \(C_0\) be such an initial configuration. First observe that even after sniping we then have at least \(2m\) agents left. Intuitively the argument will be along the following lines: Case 1: If an agent has not interacted in a copy \(j\) yet, then sniping it falls into the initial tolerance. Case 2: If an agent has value \(0\) in a copy, it can also be sniped without issue. Case 3: Neither case 1 nor 2 is the case. We have to show that for every agent, case 3 happens for at most one copy. Then the number of copies damaged is at most the number of sniped agents, i.e.\ at most \(m-1\), and the majority of copies remains correct.

Hence let \(\rho=(C_0, C_1, \dots)\) be a fair $k$-execution with \(k = \intol(C_0)\). For every agent \(a\) we define 
\begin{align*}
    J(a,n) &:= \{j \in [1,2m] :\\
    &\qquad a \text{ has interacted in the $j$-th copy of $\Z^m$ before step } n\\
    &\qquad \land v_j(a)\neq 0\}\\
    L(a, n) &:= \begin{cases}
        \{j\} & \text{if $a$ is a leader in the $j$-th copy of $\Z^m$ at step $n$} \\
        \emptyset & \text{else}
    \end{cases}
\end{align*}
    We claim that if \(C_n\) fulfills that for all agents \(J(a,n) \subseteq L(a,n)\), then also \(C_{n+1}\) does. Comparing with the outline, this states exactly that case 3 occurs only for the copy \(j\) in which agent \(a\) is the leader. To see that this property is actually inductive, we inspect every transition: \TraRef{steal} and \TraRef{retire} set the value of agent \(a\) in a copy \(j\) to \(0\) on the non-leader, \TraRef{result} does not influence these values. \TraRef{distrib} can only be taken if $v_j(a) \neq 0$ and we thus have \(J(a,n+1) = \emptyset\). Importantly, the condition is in fact even preserved by snipes, since properties of the form ``For all agents \dots'' become easier by removing agents. Additionally note that the condition holds in any initial configuration.

Next we prove that cases 1 and 2 are non-problematic. We first claim that \(\rho\) ends in a terminal configuration.

Proof of claim: \Wlog{} we reorder \(\rho\) s.t.\ \(\rho\) is of the form $C_0 \tostar C \snstar D \tostar C_{\mathit{term}}$, i.e.\ such that all snipes happen at once and it reaches a terminal configuration. To ensure this, we commute the snipes past other transitions towards the last snipe. This is possible since a non-interacting agent might as well not exist. Transitions \TraRef{steal} and \TraRef{retire} reduce the total number of non-zero entries in the copies of \(\Z_m\) and this number cannot be increased. Hence they become permanently disabled eventually. Then all non-leaders have all entries equal to \(0\), and \TraRef{distrib}  is hence also disabled. Now \TraRef{result} provides every agent with some output for every subprotocol, and afterwards becomes disabled.

Hence \(\rho\) is of the form $C_0 \tostar C \snstar D \tostar C_{\mathit{term}}$. Let \(j\) be one of the copies such that no leader for this copy was sniped (i.e.\ such that no case 3 occurred for this copy). We prove that the projection \(\pr_j\) of \(\rho\) to the \(j\)-th subprotocol has the correct output. For all agents which were sniped before ever interacting, we assume they were sniped at the start, and hence consider the initial configuration \(C_0'\) with those agents removed. \(C_0'\) still has at least \(2m\) agents. Consider the execution \(\rho'\) starting at \(C_0'\) and otherwise mirroring \(\pr_j(\rho)\), which is possible since the removed agents never interacted in the \(j\)-th subprotocol.  We prove that the \(j\)-th subprotocol has the correct output in \(\rho'\), and hence also in \(\rho\).

By choice of \(j\), the only agents which are sniped in \(\rho'\) fulfilled \(v(j)=0\), and were non-leaders for this copy. In particular after the first \TraRef{distrib}, every configuration in \(\rho'\) always had at least one leader, in particular the terminal configuration \(\pr(C_{\mathit{term}})\). Since the configuration is terminal, we have exactly one leader. The \(j\)-th sum $\Inv_j(C) := \sum_{q=(i,v,r)\in Q} v(j) \cdot C(q)$ is invariant modulo \(m\), and is not impacted by sniping of agents with \(v(j)=0\).  Hence we in particular have \(\Inv_j(\pr(C_{\mathit{term}})) \equiv \Inv_j(C_0') \mod m\). In \(\pr(C_{\mathit{term}})\) this value is fully stored by the leader of the \(j\)-th subprotocol, since otherwise \TraRef{steal} would be enabled. Hence the leader has the correct output, and therefore using transition \TraRef{result} every agent eventually has the correct output in the \(j\)-th subprotocol.

We have therefore seen that only copies \(j\) where a leader gets sniped are damaged. Since every agent is leader in at most one subprotocol, and at most \(\intol(C_0)\) many agents are sniped, it follows that at most \(k \leq m-1\) subprotocols have the wrong output. The majority \(m+1\) of subprotocols hence remain correct. Therefore the whole protocol has the correct output.
\end{proof}

%Since we need to detect that a certain number of agents exist for this protocol to be correct, we
%will later have to also handle smaller inputs. To detect if we are able to use the protocol described above, we use a tower, which
%simultaneously computes the following threshold predicate
%$\psi$: \[\psi(x) = \bigg(\sum_{i=1}^n x_i \geq 2m\bigg)\]

\subsection{Handling Small Inputs}
\label{subsec:small-modulo}

Sensei is not fully satisfied with her modulo protocol. At the moment the ninjas have a lot of members (thanks to the many new rookies), but that can change at any moment.
Therefore it is important for the modulo protocol to also handle small inputs. She thinks of the following idea: a tower computes the predicate $\psi':=\sum_{i=1}^n x_i \geq 2m$ to detect if there
are enough ninjas to successfully use the $\BigModulo$ protocol. If there are not enough ninjas, she would want to determine the exact modulo value by the highest ninja in the tower, which is not correct with this choice of \(\psi'\), since the coefficients were ignored. Hence instead she uses \(\psi:=\sum_{i=1}^n a_i x_i \geq 3m^2\), which is a sufficient condition that there are at least \(3m\) ninjas, though not a necessary condition.

Since the ninjas do not know how many of them will show up, in particular they do not know whether their number will be large or small. We solve this problem by suitably combining the protocols for the two cases. Intuitively, the ninjas execute both protocols \emph{simultaneously}, that is, at every moment in time they are in a state for each protocol.  Further, each ninja maintains an estimate of their remainder modulo $m$. Formally, we let $\Pcal_T = (Q_T, \delta_T, \_, O_T) = \InhomTower$ for the predicate $\psi$ and $\Pcal_M = (Q_M, \delta_M, \_, O_M) = \BigModulo$. The \ModuloCombined{} population protocol is defined as follows:

\begin{protocol}{\ModuloCombined}
\States{$Q = Q_T \times [0, 3m^2] \times Q_M$}
\Input{$I = \{([0, a_i), a_i, (0, a_i \cdot \one{[1, 2m]}, 0)): 1\leq i \leq n\}$}
\Output{$O((q_T, h, q_M)) = \begin{cases}
    h \bmod m \geq t & \text{if } h < 3m^2 \\
    O_M(q_M) & \text{else}
\end{cases}$}
\Transitions{}
\Done
\begin{itemize}
\item For all $(q_1, q_2 \mapsto q_1', q_2') \in \delta_T$ and $(p_1, p_2 \mapsto p_1', p_2') \in \delta_M$ and $h_1, h_2 \in [0, 3m^2]$ we add the following transition:
\[(q_1, h_1, p_1), (q_2, h_2, p_2) \mapsto (q_1', h, p_1'), (q_2', h, p_2')\TraName{parallel}\]
Where $h = \max\{h_1, h_2, \textsc{end}(q_1'), \textsc{end}(q_2')\}$ with $\textsc{end}([s, e)) = e$.
\end{itemize}
\end{protocol}

\begin{theorem}
	\ModuloCombined{} computes $\phi$ and is robust.
\end{theorem}
\begin{proof}
    Let \(C_0\) be an initial configuration, and \(\rho=(C_0,C_1, \dots)\) a fair $k$-execution with \(k = \intol(C_0)\). Since every fair execution of both \BigModulo{} and \InhomTower{} terminates, every fair execution of \ModuloCombined{} terminates. Let \(C_r\) be the terminal configuration reached. Since \TraRef{parallel} in particular updates the values of \(h_1\) and \(h_2\) to the maximum, the fact that \(C_r\) is terminal implies that every agent has the same value \(h\) in this component. Since we chose the threshold of $\psi$ to be $3m^2$, which is a multiple of $m$, we have that $\intol_{\Pcal_T}(C_0) \geq \intol_{\Pcal_M}(C_0)$. By this argument the following case distinction does not depend on snipes:
	
	\parag{Case 1: \(h<3m^2\)} Then no agent in the tower component is at the top, and the output of \ModuloCombined{} is determined by \(O_T\). 
	Let \(C_0'\) be the initial configuration if the agents sniped in \(\rho\) had been sniped at the start.
	 The height \(h\) has to be in the interval
	\(\big[\sum_{q \in Q} \len(q_T) \cdot C_0'(q), \min(\sum_{q \in Q} \len(q_T) \cdot C_0(q), 3m^2)\big]\). By definition of initial tolerance, in fact \(\varphi\)
	is constant even on the larger interval \(\big[\sum_{q \in Q} \len(q_T) \cdot C_0'(q), \sum_{q \in Q} \len(q_T) \cdot C_0(q)\big]\).
	Hence no matter which of these values \(h\) actually has, using it we obtain the correct value of \(\varphi\).

	\parag{Case 2: \(h=3m^2\)} Then the output is fully given by $O_M(q_M)$. Since \(h=3m^2\), and every coefficient \(a_i\) is at most \(m\), initially we must have had at least \(3m\) agents. Hence \(C_0\) is a large enough initial configuration that we can use
	\pref{prop:big-modulo-robust} which states $\BigModulo$ is robust for \(C_0\), 
	wherefore the projection of \(\rho\) to \(Q_M\) and therefore also \(\rho\) itself gives the correct output.
\end{proof}

\section{Future Work: Boolean Combinations}
\label{sec:boolean-comb}

Recall that every Presburger predicate can be represented as a boolean combination of threshold and modulo predicates. 
Given protocols for two predicates $\varphi_1, \varphi_2$, there is a simple construction that yields protocols deciding $\varphi_1 \wedge \varphi_2$ or $\varphi_1 \vee \varphi_2$. Intuitively, agents execute both protocols simultaneously. In particular, their states are pairs of states of the protocols for $\varphi_1$ and $\varphi_2$, and their output function is the conjunction of disjunction of the output functions for $\varphi_1$ and $\varphi_2$ \cite{AngluinADFP06}. If this construction would preserve robustness, that is, if the protocol for $\varphi_1 \wedge \varphi_2$ or $\varphi_1 \vee \varphi_2$ would be robust whenever the protocols for $\varphi_1$ and $\varphi_2$ are, then we would have proved that every predicate has a robust protocol.

Unfortunately, the construction does not preserve robustness. While applying the construction for conjunction to our \InhomTower{} protocol from \sref{subsec:inhomogeneous} yields a robust protocol, this is not the case for the \InhomTowerCancel{} protocol from \sref{subsec:combination}.

Consider an initial configuration where the first protocol accepts and the second protocol rejects. Additionally, via sniping within the initial tolerance, it is possible to make the first protocol reject and the second accept. Now the sniper can wait until the first protocol accepts and then snipe such that the second protocol also accepts.
\begin{example}
    We consider the two predicates $\phi_1(x, y) = (2x - y \geq 3)$ and $\phi_2(x, y) = (y - x \geq 1)$ and their conjunction $\phi(x, y) = (\phi_1(x, y) \land \phi_2(x, y))$.
    The reachable states of our \InhomTowerCancel{} protocol instantiated with the two predicates are:
    \begin{align*}
        Q_1 = \{[0, 2), [1, 3), -1, 0\}\qquad Q_2 = \{[0, 1), -1, 0\}
    \end{align*}
    Where the initial states corresponding to $x$ are $[0, 2)$ for $\phi_1$ and $-1$ for $\phi_2$. For $y$ we have $-1$ and $[0, 1)$ respectively.

In order to show that the construction is not robust, we search for an initial configuration that is rejected by $\phi_1$ and accepted by $\phi_2$, and which can be sniped such that $\phi_1$ accepts and $\phi_2$ rejects. From such a configuration we can have an execution that first stabilizes in the second protocol, then a snipe follows, and then the first protocol also stabilizes, leading to an overall accepting  configuration that should have been rejected. To find such a configuration, we draw the plot in \fref{fig:accept-plot}. Colors indicate the values of both predicates (see the legend of the figure). Sniping corresponds to moving left or down in the plot, and so we search for a blue square from which we can reach a red square by moving down or left, and from which we cannot reach a black square with the same number of moves. This last condition corresponds to the tolerance condition: the sniper is only allowed to snipe when the output of the initial configuration does not change.
    \begin{figure}[t]
        \centering
    \begin{tikzpicture}[scale=0.5]
    	\tikzstyle{phione} = [preaction={fill, nicebgbluebright},pattern color=nicebluebright,pattern={Lines[angle=45,distance=5pt,line width=0.9pt]}];
    	\tikzstyle{phitwo} = [preaction={fill, nicebgredbright },pattern color=niceredbright,pattern={Lines[angle=-45,distance=5pt,line width=0.9pt]}];
    	\tikzstyle{phiall} = [preaction={fill, black!15},pattern color=black,pattern={Hatch[angle=-45,distance=5pt,line width=0.9pt],}];
        \foreach \x in {1,...,7} {
            \draw (\x, -0.1) -- (\x, 0.1) node[below=0.2cm] {$\x$};
            \draw (-0.1, \x) -- (0.1, \x) node[left=0.2cm] {$\x$};
            \foreach \y in {1,...,7} {
                \pgfmathtruncatemacro{\phiA}{2*\x - \y}
                \pgfmathtruncatemacro{\phiB}{\y - \x}
                \ifnum\phiA<3
                    \ifnum\phiB<1
                        % nothing
                    \else
                        \fill[phione] 
                        (\x-0.5,\y-0.5) rectangle (\x+0.5,\y+0.5);
                    \fi
                \else
                    \ifnum\phiB<1
                        \fill[phitwo] (\x-0.50,\y-0.502) rectangle (\x+0.50,\y+0.50);
                    \else
                        \fill[phiall] (\x-0.50,\y-0.50) rectangle (\x+0.50,\y+0.50);
                    \fi
                \fi
            }
        }
        \draw[step=1,black,xshift=5mm,yshift=5mm,line width=0.4pt] (0,0) grid (7,7);
        
        \draw[->] (0,0) -- (7.5,0) node[right] {$x$};
        \draw[->] (0,0) -- (0,7.5) node[above] {$y$};
    \end{tikzpicture}   
        \caption{Plot of the predicates $\phi_1(x, y) = (2x - y \geq 3)$ and $\phi_2(x, y) = (y - x \geq 1)$ on initial configurations. The color indicates the value of the predicates: it is red if $\phi_1$ accepts, blue if $\phi_2$ accepts, black if both accept and white if both reject.}
        \label{fig:accept-plot}
    \end{figure}
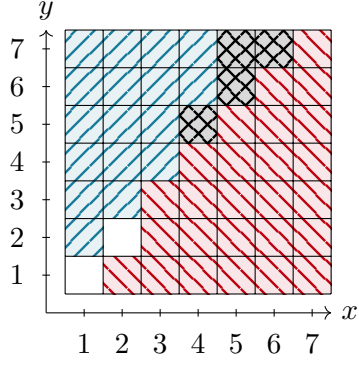
    
    The plot shows that $C = \cfg{3\cdot x, 4\cdot y}$ is a suitable configuration. Its initial tolerance is $6$, since all configurations $D \leq C$ are rejected. We now consider the following fair $1$-execution (state changes are highlighted in bold):
    \begin{align}
        &\cfg{3\cdot ([0, 2), -1), 4\cdot(-1, [0, 1))}\\
        & \qquad\tostar \cfg{3\cdot ([0, 2), \mathbf{0}), 3\cdot(-1, \mathbf{0}), \mathbf{1}\cdot(-1, [0, 1))}\\
        & \qquad\sn     \cfg{3\cdot ([0, 2), 0), \mathbf{2}\cdot(-1, 0), (-1, [0, 1))}\\
        & \qquad\tostar \cfg{\mathbf{1}\cdot([0, 2), 0), 2\cdot(\mathbf{[1, 3)}, 0), 2\cdot(-1, 0), (-1, [0, 1))}\\
        & \qquad\tostar \cfg{([0, 2), 0), 2\cdot([1, \mathbf{2}), 0), 2\cdot(\mathbf{0}, 0), (-1, [0, 1))}\\
        & \qquad\tostar \cfg{([0, 2), 0), 2\cdot(\mathbf{[2, 3)}, 0), 2\cdot(0, 0), (-1, [0, 1))}\\
        & \qquad\tostar \cfg{([0, 2), 0), \mathbf{1}\cdot([2, 3), 0), \mathbf{3}\cdot(0, 0), (\mathbf{0}, [0, 1))}
    \end{align}
    From $(1)$ to $(2)$ only the second protocol executes a non-silent transition: three of the $(-1, [0, 1))$ agents annihilate with $([0, 2), -1)$ in the second protocol. Then in $(3)$, we proceed to snipe a single $(-1, 0)$ agent. This does not influence the second protocol, since it already is finished. In $(4)$, the $([0, 2), 0)$ agents move upwards in the first protocol and after that in $(5)$, we annihilate $([2, 3), 0)$ with $(-1, 0)$ twice. In $(6)$ we again move the agents in the first protocol upwards and in $(7)$ we do a final annihilation.

 The last configuration is accepting and terminal, thus we have shown that the product construction applied to our protocol is not robust.
\end{example}
This example only shows that the simple construction for conjunction does not preserve robustness. But there could be a more sophisticated  construction, or a procedure that directly constructs a robust protocol for any predicate, for example from a finite automaton for the predicate. Currently, the simplest predicate for which we have not been able to find a robust protocol is 
\[\phi(x, y) = (x \geq y \lor (x + 1 \geq y \land x \bmod 5 = 0))\]
The problematic initial configurations are of the form $\cfg{5k \cdot x, n\cdot y}$ with $5k > n$. The initial tolerance of these configurations is $5k - n$, but the robustness of the majority protocol only gives correct results for up to $5k - n - 1$ snipes.

\bibliographystyle{splncs04}
\bibliography{references}

\appendix

\end{document}